\documentclass[12pt,reqno,fleqn]{amsart}
\usepackage{latexsym}
\usepackage{amssymb}
\usepackage{amsxtra}
\usepackage{amsmath}
\usepackage{amsfonts}
\usepackage{CJK}
\usepackage[usenames]{color}
\def\BB{{\mathcal B}}
\def\REMARK#1{}
\def\REMARK#1{{\color{blue}{#1}}}
\newcommand{\cleqn}{\setcounter{equation}{0}}
\newcommand{\clth}{\setcounter{theorem}{0}}
\newcommand {\sectionnew}[1]{\section{#1}\cleqn\clth}
\newtheorem{theorem}{Theorem}[section]
\newtheorem{lemma}[theorem]{Lemma}
\newtheorem{proposition}[theorem]{Proposition}
\newtheorem{corollary}[theorem]{Corollary}
\newtheorem{definition}[theorem]{Definition}

\renewcommand{\P}{\mathcal{P}}

\newcommand{\N}{\mathcal{N}}
\renewcommand{\L}{\mathcal{L}}

\makeatletter
    
    \newcommand{\Rmnum}[1]{\expandafter\@slowromancap\romannumeral #1@}

\def\({\left(}
\def\){\right)}
\def\[{\left[}
\def\]{\right]}
\def\d{\partial}

\def\d{\partial}
\def\ep{\epsilon}
\def\La{\Lambda}
\def\la{\lambda}
\def\om{\omega}

\begin{document}
\begin{CJK*}{GBK}{song}
\title{Block type symmetry of bigraded Toda Hierarchy}
\author[Chuanzhong Li, Jingsong He,  Yucai Su]{Chuanzhong Li\S\dag, Jingsong He\S $^*$,  Yucai Su\ddag\dag}
\dedicatory {  \S Department of Mathematics, Ningbo University, Ningbo, 315211 Zhejiang, P.\ R.\ China\\
 \dag Department of Mathematics, USTC, Hefei, 230026 Anhui, P.\ R.\ China\\
 \ddag Department of Mathematics, Tongji University, Shanghai, 200092, P.\ R.\ China}

\thanks{$^*$ Corresponding author: hejingsong@nbu.edu.cn, jshe@ustc.edu.cn}
\texttt{}

\date{}

\begin{abstract}
 In this paper, we define Orlov-Schulman's  operators $M_L$, $M_R$, and then use them to
 construct  the additional
symmetries  of the bigraded Toda hierarchy (BTH). We further show
that  these additional symmetries form an interesting infinite
dimensional Lie algebra known as a Block type Lie algebra,
 whose structure theory and  representation theory have recently received much attention in literature.
By acting on two different spaces under the weak W-constraints we find in
particular two representations of this Block type Lie algebra.
\end{abstract}


\maketitle
 \tableofcontents
\allowdisplaybreaks
\sectionnew{Introduction}

The Toda lattice equation as a completely integrable system was
introduced by Toda \cite{Toda} to describe an infinite system of
masses on a line that interact through an exponential force.
Inspired by the Sato theory on  the Kadomtsev-Petviashvili (KP)
hierarchy \cite{sato,djkm}, the
two dimensional Toda hierarchy
was constructed  by  Ueno and Takasaki \cite{UT} with the
help of difference
operators and infinite dimensional Lie algebras.
The one dimensional Toda hierarchy (TH) was also studied under the
reduction condition $L+L^{-1}=M+M^{-1}$ \cite{UT} or $L=M$  on two Lax operators
$L$ and $M$. The bigraded Toda hierarchy (BTH) of $(N,M)$-type(or simply the $(N,M)$-BTH) is the generalized Toda hierarchy
whose infinite Lax matrix has $N$ upper and $M$ lower nonzero diagonals\cite{solutionBTH}. After continuous interpolation, $N$ and $M$ correspond to the highest and lowest powers of Laurent polynomials which is the Lax operator.
The $(N,M)$-BTH can be naturally considered as a reduction of the two-dimensional Toda hierarchy
by imposing an algebraic relation to those two Lax operators(see \cite{UT, solutionBTH}).
One dimensional Toda hierarchy(i.e. $(1,1)$-BTH) and the BTH have been shown to be related to many mathematical and physical fields such
as the inverse scattering method, finite and infinite dimensional
algebras, classical and quantum field theories and so on. Recently, in
\cite{CDZ,DZ virasora}, the interpolated Toda lattice
hierarchy was generalized to the so-called extended Toda hierarchy (ETH) for considering
its application on the topological field theory. In \cite{C}, the TH and ETH were further
generalized to the extended bigraded Toda hierarchy (EBTH) by considering $N+M$ dependent variables
$\{u_{N-1},u_{N-2},\cdots, u_1,u_0,u_{-1,},\cdots,u_{-M}\}$ in the Lax operator $\L$.
This new model has been expected \cite{C} that it might be relevant for applications in
describing the Gromov-Witten invariants. In fact the dispersionless case of that model has
been proposed  in \cite{CMP kodama} because the dispersionless EBTH
can be obtained from the dispersionless KP hierarchy.  In \cite{TH},
the Hirota bilinear equations (HBEs) of EBTH have been given
conjecturally and proved that it governs the Gromov-Witten theory of
orbifold $c_{km}$. In \cite{our paper}, the authors generalize the
Sato theory to the EBTH and give the Hirota bilinear equations in
terms of vertex operators whose coefficients  take values in the
algebra of differential operators. In \cite{Carlet frobenius}, a geometric structure
associated with Frobenius manifold of 2D Toda hierarchy was introduced.
Furthermore, motivated by the potential applications of the
BTH, which is also defined by omitting the extended logarithmic flows  of the
EBTH, in the theory of the matrix models, it is necessary and interesting to explore its algebraic structure from the point of
view of the additional symmetry.

Additional symmetries  of KP hierarchy were given by Orlov and
Shulman \cite{os1} through two novel operators $\Gamma$ and $M$, which can be used
to form a centerless  $W$ algebra.  Based on this work, there exist
many extensive results (e.g., [\ref{adler92}--\ref{jingsong he}])
on the additional symmetries of the KP hierarchy, Toda hierarchy, 2-D
Toda hierarchy, BKP hierarchy and CKP hierarchy. Particularly, the representations of the
infinite dimensional Virasoro algebra and $W$
algebra have been derived by using the actions on the Lax operator,
the wave function and the $\tau$ function of additional symmetry flows.
These results inspire us to search new infinite dimensional algebras from the additional symmetry flows of the
BTH. So the purpose of this paper is to give the additional symmetries of the BTH and then identify its algebraic
structure. In \cite{kodama W algebra}, the additional symmetries of KP hierarchy were generalized
to a $W_{1+\infty}$ algebra. However, the commutative relations of the $W_{1+\infty}$  algebra are rather
complicated. The algebra under consideration in this paper is very simple and elegant, which is an infinite
dimensional Lie algebra $\BB$, known as a {\it Block type Lie algebra}, introduced by Block \cite{Block} around 50 years ago.
This kind of Lie algebra is an interesting object in the structure theory and representation theory of Lie
algebras, partly due to its close relation with the Virasoro algebra and the Virasoro-like algebra
(e.g., [\ref{DZ}--\ref{Su}]).
To our best knowledge, this is the first time to bring Block type Lie algebras to the integrable
system.
In particular, we obtain two representations of the Lie algebra $\BB$
on two spaces of functions $P_L$, $P_R$ respectively (Theorem \ref{Rep-Block}), which (to the best of our knowledge) are the first known
examples of representations of $\BB$ with the actions of generating elements of $\BB$ being explicitly given.
Another quite different feature of additional symmetries in the BTH is that
$M\triangleq M_L-M_R $ operator for constructing additional flows commutes with Lax operator $\L$, i.e.,
$[\L,M]=0$. This is a crucial fact to find Block type Lie algebra here. Note that  $[L,M]=1$ holds for other known
integrable systems such as KP hierarchy.

The paper is organized as follows. In Section 2, the two dimensional Toda hierarchy and its reductions are  introduced explicitly by which we can define the BTH later. In Section 3,  the
definition of the BTH  and its Sato equation are introduced.  In Section 4, we define
Orlov-Schulman's $M_L$, $M_R$ operators and prove their linear equations.  The additional symmetries and related equations of the BTH will be given in
Section 5, meanwhile we prove that the additional symmetries have a nice structure of a Block type Lie algebra. In Section 6, we give some specific actions of that Block type additional flows on spaces of functions $P_L$, $P_R$  which further lead to representations of this Block type Lie algebra under so-called  weak W-constraints.
Section 7
is devoted to conclusions and discussions.
\section{The two-dimensional Toda hierarchy and its reductions}
In this section, we will show that the BTH is just a general reduction of the two-dimensional Toda  hierarchy whose special reduction leads to original Toda hierarchy.

Firstly we will introduce the definition of the two-dimensional Toda hierarchy in interpolated form as following.

The two-dimensional Toda hierarchy \cite{UT} can be defined by the following two  Lax operators,
\begin{align}
L&=\Lambda+a_0+ a_{-1}\La^{-1}+ a_{-2}\La^{-2}+\dots,\\[1.0ex]
\bar L&=
\bar a_{-1}\Lambda^{-1}+\bar a_0+ \bar a_{1}\La^{1}+ \bar a_{2}\La^{2}+\dots,
\end{align}
where $\Lambda$ represents the shift operator with $\Lambda:=e^{\epsilon\partial_x}$ and $``\epsilon"$ is
called the string coupling constant, i.e. for any function $ f(x)$
\begin{equation*}
\Lambda f(x)=f(x+\epsilon).
\end{equation*}
The coefficients $a_n$ and $\bar a_k$ are the functions of $x$ and $\{(x_n,y_n):n=1,2,\ldots\}$.
Then the Lax representation of the two-dimensional Toda hierarchy is given by the set of infinite number of
equations for $n=1,2,\ldots$,
\begin{align}
\frac{\d L}{\d x_n}&=[L^n_+,L],\qquad \frac{\d L}{\d y_n}=[\bar L^n_-,L],\\
\frac{\d \bar L}{\d x_n}&=[L^n_+,\bar L],\qquad \frac{\d \bar L}{\d y_n}=[\bar L^n_-,\bar L],
\end{align}
where $L^n_+$ represents the part of $L^n$ with non-negative powers in $\Lambda$, and
$\bar L^n_-$ represents the part of $\bar L^n$ with negative powers in $\Lambda$.
In particular, the Lax equations for $n=1$ provide the system for $(a_0,\bar a_{-1})$,
\begin{equation}\left\{
\begin{array}{llll}
\displaystyle{\frac{\d \bar a_{-1}(x)}{\d x_1}}&=\bar a_{-1}(x)(a_0(x)-a_0(x-\ep)),\\[1.5ex]
\displaystyle{\frac{\d  a_{0}(x)}{\d y_1}}&=\bar a_{-1}(x)-\bar a_{-1}(x+\ep).
\end{array}\right.
\end{equation}
With the function $u(x)$ defined by $ a_{0}(x)=\d_{x_1}u(x)$ and $ \bar a_{-1}(x)=e^{u(x)-u(x-\ep)} $,  the two-dimensional Toda equation,
\begin{equation}
\frac{\d^2 u(x)}{\d x_1\d y_1}=e^{u(x)-u(x-\ep)}-e^{u(x+\ep)-u(x)}
\end{equation}
is given.

The 1-D Toda equation is given by the reduction,
\begin{equation}\label{1Toda}
L=\bar L=\Lambda+a_0+a_{-1}\Lambda^{-1}.
\end{equation}
Then the Lax equations for $x_1$ and $y_1$ give
\begin{equation}
\left(\frac{\partial }{\partial x_1}+\frac{\partial}{\partial y_1}\right)L=[L_++\bar L_-,L]=[L,L]=0.
\end{equation}
This implies that $L$ does not depend on the variable $s_1:=x_1+y_1$.  Then
the two-dimensional Toda equation  is reduced to the 1-D Toda equation, and with $t_1:=x_1-y_1, \bar u(x)=-u(x)$, we have
the standard 1-D Toda equation as
\begin{equation}
\frac{\d^2 \bar u(x)}{\d t_1^2}=e^{\bar u(x-\ep)-\bar u(x)}-e^{\bar u(x)-\bar u(x+\ep)}.
\end{equation}
In the next subsection,
 we  generalize the reduction (\ref{1Toda}) to the general Laurent polynomial, i.e.
 \begin{equation}\label{NMToda}
 L^N=\bar L^M,\ \ \  \,  N,M \in \N,
 \end{equation}
 which defines the $(N,M)$-BTH($L$ and $\bar L$ will correspond to fractional powers of Lax operator of the BTH (see \eqref{frac})).

\sectionnew{ Bigraded Toda hierarchy  }
The Lax form of the BTH(i.e. $(N,M)$-BTH) can be introduced as \cite{C}. For that we need to introduce
firstly the Lax operator
\begin{eqnarray}\L=\Lambda^{N}+u_{N-1}\Lambda^{N-1}+\dots + u_{-M}
\Lambda^{-M}
\end{eqnarray}
(where $N,M \geq1$ are two fixed positive integers).
  The variables $u_j$ are functions of the real variable $x$.  The Lax operator $\L$ can be
written in two different ways by dressing the shift operator
 \begin{eqnarray}\label{constraint lax}
 \L=\P_L\Lambda^N\P_L^{-1} = \P_R \Lambda^{-M}\P_R^{-1}.
 \end{eqnarray}
 Equation \eqref{constraint lax} is  quite important because it gives the reduction condition \eqref{NMToda} of the BTH from the two-dimensional Toda hierarchy.

The two dressing operators have the following form \begin{eqnarray}
&& \P_L=1+w_1\Lambda^{-1}+w_2\Lambda^{-2}+\ldots,
\label{dressP}\\
&& \P_R=\tilde{w_0}+\tilde{w_1}\Lambda+\tilde{w_2}\Lambda^2+ \ldots,
\label{dressQ} \end{eqnarray}
and their inverses have form
\begin{eqnarray}
&& \P_L^{-1}=1+\Lambda^{-1}w'_1+\Lambda^{-2}w'_2+\ldots,
\label{dressP'}\\
&& \P_R^{-1}=\tilde{w'_0}+\Lambda\tilde{w'_1}+\Lambda^2\tilde{w'_2}+ \ldots.
\label{dressQ'} \end{eqnarray}
The coefficients $\{w_i,\tilde{w_i},w'_i, \tilde{w'_i}, i\geq 0\}$
will be used more later in calculating the representation of Block algebra.  The
pair is unique up to multiplying $\P_L$ and $\P_R$ from the right
 by operators in the form  $1+
a_1\Lambda^{-1}+a_2\Lambda^{-2}+...$ and $\tilde{a}_0 +
\tilde{a}_1\Lambda +\tilde{a}_2\Lambda^2+\ldots$ respectively with
coefficients independent of $x$. From the first identity of
\eqref{constraint lax},  the relations of $u_i, -M\leq i\leq N-1$ and $w_j, j\geq 1$ are as follows (see
\cite{our paper})
\begin{eqnarray}\label{w1with u}
  \!\!  \!\!  \!\!u_{N-1}&\!\!\!\!\!\!\!\!\!\!\!\!=\!\!\!\!\!\!\!\!\!\!\!\!&w_1(x)-w_1(x\!+\!N\epsilon), \\
  \!\!   \!\!  \!\! u_{N-2}&\!\!\!\!\!\!\!\!\!\!\!\!=\!\!\!\!\!\!\!\!\!\!\!\!&w_2(x)-w_2(x\!+\!N\epsilon)
  -(w_1(x)-w_1(x\!+\!N\epsilon))w_1(x\!+\!(N\!-\!1)\epsilon), \\ \notag
  \!\!  \!\!  \!\!   u_{N-3}&\!\!\!\!\!\!\!\!\!\!\!\!=\!\!\!\!\!\!\!\!\!\!\!\!&w_3(x)-w_3(x\!+\!N\epsilon)\\\notag
   &\!\!\!\!\!\!\!\!\!\!\!\!\!\!\!\!\!\!\!&- [w_2(x)-w_2(x\!+\!N\epsilon)
  -(w_1(x)-w_1(x\!+\!N\epsilon))w_1(x\!+\!(N\!-\!1)\epsilon)]w_1(x\!+\!(N\!-\!2)\epsilon)
  \\&&-(w_1(x)-w_1(x\!+\!N\epsilon))w_2(x\!+\!(N\!-\!1)\epsilon), \\ \notag
  \notag \cdots &\!\!\!\!\!\!\!\!\!\!\!\!\cdots\!\!\!\!\!\!\!\!\!\!\!\! &\cdots
\end{eqnarray}
Moreover, by using the second identity of \eqref{constraint lax}, we can
also easily get the relations of $u_i$ and $\tilde w_j$ formally as
follows
\begin{eqnarray}\label{w0with u}
   \!\!  \!\!  \!\! u_{-M}&\!\!\!=\!\!\!&\frac{\tilde w_0(x)}{\tilde w_0(x-M\epsilon)}, \\
  \!\!  \!\!  \!\!  \!\!  \!\!  u_{-M+1}&\!\!\!=\!\!\!&\frac{\tilde w_1(x)-\frac{\tilde w_0(x)}{\tilde w_0(x-M\epsilon)}\tilde w_1(x-M\epsilon)}{\tilde w_0(x-(M-1)\epsilon)},
  \\
   \!\!  \!\! \!\!  \!\!  \!\!   u_{-M+2}&\!\!\!=\!\!\!&\frac{\tilde w_2(x)-\frac{\tilde w_0(x)}{\tilde w_0(x-M\epsilon)}\tilde w_2(x-M\epsilon)-\frac{\tilde w_1(x)-\frac{\tilde w_0(x)}
   {\tilde w_0(x-M\epsilon)}\tilde w_1(x-M\epsilon)}{\tilde w_0(x-(M-1)\epsilon)}\tilde w_1(x-(M-1)\epsilon)}{\tilde w_0(x-(M-2)\epsilon)},
\\
  \notag \cdots &\!\!\!\cdots\!\!\! &\cdots
\end{eqnarray}
These relations above will be used in the calculation later.
Given any difference  operator $A= \sum_k A_k \Lambda^k$, the
positive and negative projections are given by $A_+ = \sum_{k\geq0}
A_k \Lambda^k$ and $A_- = \sum_{k<0} A_k \Lambda^k$.

To write out  explicitly the Lax equations  of the BTH,
  fractional powers $\L^{\frac1N}$ and
$\L^{\frac1M}$ are defined by
\begin{equation}
  \notag
  \L^{\frac1N} = \Lambda+ \sum_{k\leq 0} a_k \Lambda^k , \qquad \L^{\frac1M} = \sum_{k \geq -1} b_k
  \Lambda^k,
\end{equation}
with the relations
\begin{equation}
  (\L^{\frac1N} )^N = (\L^{\frac1M} )^M = \L.
\end{equation}
Acting on free function $\la^{\frac{x}{\ep}}$, these two fraction powers can be seen as two different locally expansions around zero and infinity respectively.
It was  stressed that $\L^{\frac1N}$ and $\L^{\frac1M}$ are two
different operators even if $N=M(N, M\geq 2)$ in \cite{C} due to two different dressing operators. They can also be
expressed as following
\begin{equation} \label{frac}
  \L^{\frac1N} = \P _{L}\Lambda\P_{L}^{-1}, \qquad \L^{\frac1M} = \P_{R}\Lambda^{-1} \P_{R}^{ -1}.
\end{equation}

Similar to \cite{C},  the BTH can be defined as following.
\begin{definition} \label{deflax}
\rm
%
%
{The  bigraded Toda hierarchy consists of a system of flows given
in the Lax pair formalism by
\begin{eqnarray}
  \label{edef}
\d_{ t_{\gamma, n}} \L = [ A_{\gamma,n}, \L ]
\end{eqnarray}
for $\gamma = N,N-1,N-2, \dots , -M+1$ and $n \geq 0$.  The
operators $A_{\gamma ,n}$ are defined by
\begin{subequations}
\label{Adef}
\begin{align}
  &A_{\gamma,n} =  (\L^{n+1-\frac{\gamma-1}N })_+ \quad \text{for} \quad \gamma = N,N-1, \dots,2, 1,\\
  &A_{\gamma,n} =  -( \L^{n+1+\frac{\gamma}M })_- \quad \text{for} \quad \gamma = 0, -1,\dots, -M+1,
\end{align}
\end{subequations}
and $\d_{ t_{\gamma, n}}$ is defined as  $\frac{\partial }{\partial t_{\gamma, n}}. $}
\end{definition}
The only difference from \cite{C} is that we cancel the extended flows and add the flow when
$\gamma=1$. The flow when
$\gamma=1$ is in fact the Toda hierarchy which is also
the flow when $\gamma=0$.

 Particularly for $N=1=M$, this hierarchy
coincides with the one dimensional Toda hierarchy.
When $N=1,M=2$, the BTH leads the following primary equations
\begin{eqnarray}
\partial_{1,0} \L= [\Lambda +u_0, \L],
\end{eqnarray}
and
\begin{eqnarray}
 \partial_{-1,0} \L= -[e^{(1+\Lambda^{-1})^{-1}\log u_{-2}}\Lambda^{-1} , \L],
\end{eqnarray}
which further lead to
\begin{eqnarray}
\begin{cases}
\partial_{1,0} u_0(x)= u_{-1}(x+\epsilon)-u_{-1}(x),\\
\partial_{1,0} u_{-1}(x)= u_{-2}(x+\epsilon)-u_{-2}(x)+u_{-1}(x)(u_0(x)-u_0(x-\epsilon)),\\
 \partial_{1,0} u_{-2}(x)=u_{-2}(x)( u_0(x)-u_0(x-2\epsilon)),
  \end{cases}
\end{eqnarray}
and
\begin{eqnarray}\label{-1,0flow}
\begin{cases}
\partial_{-1,0} u_0(x)= e^{(1+\Lambda^{-1})^{-1}\log u_{-2}(x+\epsilon)}-e^{(1+\Lambda^{-1})^{-1}\log u_{-2}(x)},\\
\partial_{-1,0} u_{-1}(x)= e^{(1+\Lambda^{-1})^{-1}\log u_{-2}(x)}(u_0(x)-u_0(x-\epsilon)),\\
  \partial_{-1,0} u_{-2}(x)=u_{-1}(x)e^{(1+\Lambda^{-1})^{-1}\log u_{-2}(x-\epsilon)}-e^{(1+\Lambda^{-1})^{-1}\log u_{-2}(x)} u_{-1}(x-\epsilon).
  \end{cases}
\end{eqnarray}
Obviously equation \eqref{-1,0flow} contains infinite multiplication because of nonlocal term $(1+\Lambda^{-1})^{-1}\log u_{-2}(x)$ which comes from the fractional power of the Lax operator.

Set $N=2$ and $M=1$, the equations \eqref{edef} are as follows
\begin{eqnarray}
&& \partial_{2,0} \L= [\Lambda +(1+\La)^{-1}u_1(x), \L],
\\
&&\partial_{1,0} \L= [\Lambda^2+u_1\Lambda +u_0, \L],
\end{eqnarray}
which further lead to the following concrete equations
\begin{eqnarray}\label{N=2,M=12,0flow}
\begin{cases}
 \partial_{2,0} u_1(x)= u_1(x+\epsilon)-u_1(x)+u_1(x)(1-\La)(1+\La)^{-1}u_1(x),\\
\partial_{2,0} u_0(x)= u_{-1}(x+\epsilon)-u_{-1}(x),\\
\partial_{2,0} u_{-1}(x)= u_{-1}(x)(1-\La^{-1})(1+\La)^{-1}u_1(x),
 \end{cases}
\end{eqnarray}

\begin{eqnarray}\label{N=2,M=11,0flow}
\begin{cases}
\partial_{1,0} u_1(x)= u_{-1}(x+2\epsilon)-u_{-1}(x),\\
\partial_{1,0} u_0(x)= u_{-2}(x+2\epsilon)-u_{-2}(x)+u_1(x)u_{-1}(x+\epsilon)-u_{-1}u_1(x-\ep),\\
\partial_{1,0} u_{-1}(x)= u_{-1}(x)(u_0(x)-u_0(x-\epsilon)).
 \end{cases}
\end{eqnarray}
Notice that the nonlocal term $(1+\La)^{-1}u_1(x)$ also comes from the fractional power of Lax operator in equations above. Therefore appearance of nonlocal term is an important property of the BTH.
We can also get more equations when $N$ and  $M$ take other integer values but we shall
not mention them here because our central consideration in
this paper is the Block type additional symmetries of the BTH.

 For the convenience to derive the Sato equations,  the following operators will be
defined as in \cite{C,our paper}:

\begin{eqnarray}
  B_{\gamma , n} :=
\begin{cases}
  \L^{n+1-\frac{\gamma-1}{N}} &\mbox{for \ }\gamma=N\dots 1,\\
  \L^{n+1+\frac{\gamma}{M}} &
  \mbox{for \ }\gamma = 0\dots -M+1.
  \end{cases}
\end{eqnarray}
Before introducing the Sato equation, the following proposition \cite{C} need to be given firstly.
\begin{proposition} \label{lemD1}
The following two identities hold
\begin{align}
  \label{d6}
   & \d_{t_{\gamma,n}}\L^{\frac1N} = [ - (B_{\gamma,n})_-, \L^{\frac1N} ], \\
  \label{d6i} &\d_{t_{\gamma,n}}\L^{\frac1M} = [ (B_{\gamma,n})_+, \L^{\frac1M} ].
\end{align}
\end{proposition}
\begin{proof} See \cite{C,our paper}.\end{proof}

 Using the proposition above, one can obtain the following
proposition, lemma and theorem, which are results of \cite{C, our
paper}.

\begin{proposition} \label{propZS}
If $L$ satisfies the Lax equation \eqref{edef}, then we have the following
Zakharov-Shabat equation
\begin{eqnarray}
\label{zs}  \d_{t_{\beta,n}}(A_{\alpha,m}) - \d_{t_{\alpha,m}} (A_{\beta,
n})+ [ A_{\alpha,m} , A_{\beta,n} ] =0,
\end{eqnarray}
for $-M+1 \leq \alpha, \beta \leq N$ ,  $m, n \geq 0$.
\end{proposition}

 Using the Zakharov-Shabat equation (\ref{zs}) one can obtain the following lemma.
\begin{lemma} \label{zs corollary1} {\rm~(\cite{our paper})~}
The following Zakharov-Shabat equations hold
\begin{align}\label{compatible zero curvature}
 &\partial_{\beta,n}(B_{\alpha,m})_{-} - \partial_{\alpha,m}(B_{\beta,
n})_{-} - [ (B_{\alpha,m})_{-} , (B_{\beta,n})_{-} ] =
0,\\
 &-\partial_{\beta,n}(B_{\alpha,m})_{+} + \partial_{\alpha,m}(B_{\beta,
n})_{+} - [ (B_{\alpha,m})_{+} , (B_{\beta,n})_{+} ] =0,
\end{align}
where, $-M+1\leq\alpha,\beta\leq N$,  $m,n \geq 0$.
\end{lemma}

Using Lemma \ref{zs corollary1} and the Lax equation, one can then
obtain the following theorem.

\begin{theorem}
\label{t1} {\rm~(\cite{our paper})~} $L$ is a solution to the BTH
if and only if there is a pair of dressing operators $\P_L$ and
$\P_R$, which satisfy the following Sato  equations:
{\allowdisplaybreaks}
\begin{eqnarray}
\label{bn1}
\d_{\gamma,n}\P_L & =&- (B_{\gamma,n})_-  \P_L, \\
\label{bn1'} \d_{\gamma,n}\P_R & =& (B_{\gamma ,n})_+\P_R,
\end{eqnarray}
where, $-M+1\leq\gamma\leq N$,  $n \geq 0.$
\end{theorem}
The dressing operators satisfying Sato equations \eqref{bn1} and \eqref{bn1'} will be called wave operators later.

After the preparation above, it is time to introduce  Orlov-Schulman's operators which is included in the next section.

\section{Orlov-Schulman's $M_L$, $M_R$ operators}
 In order to give the additional symmetries of the BTH, we define the Orlov-Schulman's $M_L$, $M_R$
 operators
 by
\begin{eqnarray}\label{Moperator}
&&M_L=\P_L\Gamma_L \P_L^{-1}, \ \ \ \ \ \ \ M_R=\P_R\Gamma_R
\P_R^{-1},
\end{eqnarray}
where
\begin{eqnarray}
 && \Gamma_L=
\frac{x}{N\epsilon}\Lambda^{-N}+\sum_{n\geq 0}\sum_
{\alpha=1}^{N}(n+1 -
\frac{\alpha-1}{N} )
\Lambda^{N({n-\frac{\alpha-1}{N}})}t_{\alpha, n},
\\
&&
\Gamma_R=-\frac{x}{M\epsilon}\Lambda^M-\sum_{n\geq 0}\sum_
{\beta=-M+1}^{0}(n+1 +
\frac{\beta}{M} )
\Lambda^{-M({n+\frac{\beta}{M}})}t_{\beta, n}.
\end{eqnarray}
A direct calculation shows that the operators $M_L$ and $M_R$
satisfy the following theorem.
\begin{proposition}\label{flowsofM} Operators $\L$,$M_L$ and $M_R$
satisfy the following identities
\begin{eqnarray}\label{LM=1}
[\L,M_L]=1,&&[\L,M_R]=1,\\\label{LM time} \partial_{ t_{\gamma,n}}M_L=
[A_{\gamma,n},M_L],& &\partial_{ t_{\gamma,n}}M_R=[A_{\gamma,n},M_R],
\\\label{LMmix t}
\dfrac{\partial M_L^n\L^k}{\partial{t_{\gamma,n}}}=[A_{\gamma,n},
M_L^n\L^k],&& \dfrac{\partial
M_R^n\L^k}{\partial{t_{\gamma,n}}}=[A_{\gamma,n}, M_R^n\L^k],
\end{eqnarray}
where,
$ -M+1\leq \gamma \leq N, n\geq 0.$
\end{proposition}
\begin{proof}
Firstly we prove \eqref{LM=1} by dressing the following two identities using $\P_L$ and $\P_R$ separately

\begin{eqnarray*}[\Lambda^{N},\Gamma_L]&=&[\Lambda^{N}, \frac{x}{N\epsilon}\Lambda^{-N}+\sum_{n\geq 0}\sum_
{\alpha=1}^{N}(n+1-\frac{\alpha-1}{N}) \Lambda^{N({n-\frac{\alpha-1}{N}})}t_{\alpha,
n}]\\&=&1,
\\[5pt]
%
%
[\Lambda^{-M},\Gamma_R]&=&[\Lambda^{-M},-\frac{x}{M\epsilon}\Lambda^M-\sum_{n\geq 0}\sum_ {\beta=-M+1}^{0}(n+1 + \frac{\beta}{M}
)
\Lambda^{-M({n+\frac{\beta}{M}})}t_{\beta,
n}]\\&=&1.\end{eqnarray*} For the proof of \eqref{LM time}, we need
to prove
\begin{eqnarray}\label{MLalpha}
&&\partial_{ t_{\alpha,n}}M_L=
[(B_{\alpha,n})_{+},M_L],\\\label{MRalpha}
&&\partial_{ t_{\alpha,n}}M_R=[(B_{\alpha,n})_{+},M_R],\\\label{MLbeta}
&&\partial_{ t_{\beta,n}}M_L=[-(B_{\beta,n})_{-},M_L],\\\label{MRbeta}
&&\partial_{ t_{\beta,n}}M_R=[-(B_{\beta,n})_{-},M_R],
\end{eqnarray}
where, $1\leq \alpha \leq N, -M+1\leq \beta \leq 0.$

Let us consider the following bracket
\begin{eqnarray}\label{gammaLbracket}[\partial_{ t_{\alpha,n}}-\Lambda^{N({n+1-\frac{\alpha-1}{N}})},\Gamma_L]=0,
\end{eqnarray}
which can be easily got by a direct computation. By dressing both
sides of the identity above using the operator $\P_L,$  the left
part of \eqref{gammaLbracket} becomes
\begin{eqnarray*}&&[\P_L\partial_{ t_{\alpha,n}}\P_L^{-1}-\L^{n+1-\frac{\alpha-1}{N}},\P_L\Gamma_L\P_L^{-1}]
=
[\partial_{t_{\alpha,n}}-(B_{\alpha,n})_{+},M_L],
\end{eqnarray*}which leads to \eqref{MLalpha}.

For the proof of \eqref{MLbeta}, we consider the following bracket
\begin{eqnarray}[\partial_{ t_{\beta,n}},\Gamma_L]=0.
\end{eqnarray}
By dressing the identity above in the same way we can get
\begin{eqnarray*}[\P_L\partial_{ t_{\beta,n}}\P_L^{-1},\P_L\Gamma_L\P_L^{-1}]=[\partial_{ t_{\beta,n}}+(B_{\beta,n})_{-},M_L],\end{eqnarray*}
 which leads to
$$\partial_{ t_{\beta,n}}M_L=[-(B_{\beta,n})_{-},M_L].$$
Similarly, we can prove  \eqref{MRalpha} and \eqref{MRbeta} by
dressing the identity
\begin{eqnarray*}[\partial_{ t_{\alpha,n}},\Gamma_R]=0,\end{eqnarray*}
and \begin{eqnarray*}[\partial_{ t_{\beta,n}}+\Lambda^{-M({n+1+\frac{\beta}{M}})},\Gamma_R]=0,
\end{eqnarray*}
respectively, through the dressing operator $\P_R$. Using proved  \eqref{LM time} and Lax equation
\eqref{edef}, we can prove \eqref{LMmix t} easily. \end{proof}

The equations in Proposition \ref{flowsofM} can be realized by linear equations in the
following proposition. To simplify the  theorem, we first introduce two
functions
 $w_L(t,\la)$ and $w_R(t,\la)$ which  have forms
\begin{eqnarray}\label{PLsymbol}
w_L(t,\la) &=&\P_L(x,\Lambda)e^{\xi_L(t,\la)}=P_L(x,\Lambda)e^{\xi_L(t,\la)},\\\label{PRsymbol}
w_{R}(t,\la) &=&\P_R(x,\Lambda)e^{\xi_R(t,\la)}=P_R(x,\Lambda)e^{\xi_R(t,\la)},
\end{eqnarray}
where
\begin{eqnarray}
\xi_L(t,\la) &=&\sum_{n\geq 0}\sum_
{\alpha=1}^{N}\la^{({n+1-\frac{\alpha-1}{N}})}t_{\alpha, n}
+\frac{x}{N\epsilon}\log \la,\\
\xi_{R}(t,\la) &=&-\sum_{n\geq 0}\sum_
{\beta=-M+1}^{0}
\la^{-({n+1+\frac{\beta}{M}})}t_{\beta,
n}+\frac{x}{M\epsilon}\log \la.
\end{eqnarray}
Functions $w_L(t,\la)$ and $w_{R}(t,\la)$ will be called wave functions.  $P_L$ and $P_R$ are called symbols of $\P_L$ and $\P_R$ respectively.
\begin{proposition}
The wave functions $w_L(t,\la)$ and $w_R(t,\la)$  satisfy the following linear equations
\begin{eqnarray}
\begin{cases}\label{wLlinear}
\L w_L(t,\la)=&\la w_L(t,\la),\\
M_Lw_L(t,\la)=&\partial_\la w_L(t,\la),\\
 \partial_{ t_{\gamma,n}}w_L(t,\la)=&A_{\gamma,n}w_L(t,\la),
 \end{cases}
 \end{eqnarray}
 \begin{eqnarray}
 \begin{cases}\label{wRlinear}
 \L w_{R}(t,\la)=&\la^{-1}w_{R}(t,\la),\\
M_{R}w_{R}(t,\la)=&\partial_{\la^{-1}}w_{R}(t,\la),\\
\partial_{
t_{\gamma,n}}w_{R}(t,\la)=&A_{\gamma,n}w_{R}(t,\la),
 \end{cases}
\end{eqnarray}
where, $ -M+1  \leq \gamma \leq N, n\geq 0$.

\end{proposition}
\begin{proof}
Because
\begin{eqnarray*}
\Lambda^N\exp(\frac{x}{N\epsilon}\log \la)=\la\exp(\frac{x}{N\epsilon}\log \la),
\end{eqnarray*}
so
\begin{eqnarray*}
\Lambda^Ne^{\xi_L(t,\la)}=\la e^{\xi_L(t,\la)}.
\end{eqnarray*}
Then using the definition of dressing operators, we can get
\begin{eqnarray*}\L w_L(t,\la)
&=&\L\P_L(x,\Lambda)e^{\xi_L(t,\la)}\\
& =&\P_L(x,\Lambda)\Lambda^Ne^{\xi_L(t,\la)}\\
&=&\la w_L(t,\la).
\end{eqnarray*}
Similarly we can do the following computation
\begin{eqnarray*} M_Lw_L(t,\la)
&=&M_L\P_L(x,\Lambda)e^{\xi_L(t,\la)}\\
&=&\P_L(x,\Lambda)(\frac{x}{N\epsilon}\Lambda^{-N}+\sum_{n\geq
0}\sum_ {\alpha=1}^{N}(n+1-\frac{\alpha-1}{N} )\Lambda^{N({n-\frac{\alpha-1}{N}})}t_{\alpha,
n})e^{\xi_L(t,\la)}\\
&=&\P_L(x,\Lambda)(\frac{x}{N\epsilon}\la^{-1}+\sum_{n\geq 0}\sum_
{\alpha=1}^{N}(n+1-\frac{\alpha-1}{N} )\la^{{n-\frac{\alpha-1}{N}}}t_{\alpha, n})e^{\xi_L(t,\la)}\\
&=&\P_L(x,\Lambda)\frac{\d}{ \d \la}e^{\xi_L(t,\la)}\\
&=&\frac{\d}{\d \la}w_L(t,\la).
\end{eqnarray*}

For the time flows of linear functions, we have to consider flows when $1\leq \gamma \leq N$  and flows when $ -M+1\leq \gamma  \leq 0$  separately.
Setting $1\leq \alpha \leq N, -M+1\leq \beta \leq 0,$ and
considering  the relations
\begin{eqnarray*}
e^{\xi_L(t,\la) }&=&\exp(\sum_{n\geq 0}\sum_
{\alpha=1}^{N}\La^{N({n+1-\frac{\alpha-1}{N}})}t_{\alpha, n})
e^{\frac{x}{N\epsilon}\log \la},
\end{eqnarray*}
we do the derivative as follows
\begin{eqnarray*}\!\!\!\! &&\d_{\alpha,n}w_L(t,\la)\\
\!\!\!\!&=&(\d_{\alpha,n}\P_L(x,\Lambda))\exp(\sum_{n\geq 0}\sum_
{\alpha=1}^{N}\La^{N({n+1-\frac{\alpha-1}{N}})}t_{\alpha, n})
e^{\frac{x}{N\epsilon}\log \la}\\
\!\!\!\!&&+\P_L(x,\Lambda)\La^{N({n+1-\frac{\alpha-1}{N}})}\exp(\sum_{n\geq 0}\sum_
{\alpha=1}^{N}\La^{N({n+1-\frac{\alpha-1}{N}})}t_{\alpha, n})
e^{\frac{x}{N\epsilon}\log \la}\\
\!\!\!\!&=&-(B_{\alpha,n})_-\P_L(x,\Lambda)e^{\xi_L(t,\la) }+\L^{{n+1-\frac{\alpha-1}{N}}}\P_L(x,\Lambda)e^{\xi_L(t,\la) }\\
\!\!\!\!&=&(B_{\alpha,n})_+w_L(t,\la).
\end{eqnarray*}
Because $\xi_L(t,\la)$ does not depend  the time variables $t_{\beta,n}$, we have
%
%
%
\begin{eqnarray*} \d_{\beta,n}w_L(t,\la)
&=&(\d_{\beta,n}\P_L(x,\Lambda))\exp(\sum_{n\geq 0}\sum_
{\alpha=1}^{N}\La^{N({n+1-\frac{\alpha-1}{N}})}t_{\alpha, n})
e^{\frac{x}{N\epsilon}\log \la}\\
&=&-(B_{\beta,n})_-\P_L(x,\Lambda)e^{\xi_L(t,\la) }\\
&=&-(B_{\beta,n})_-w_L(t,\la).
\end{eqnarray*}
Now the proof of \eqref{wLlinear} is finished. In the similar way, using
\begin{eqnarray*}
\Lambda^{-M}\exp(\frac{x}{M\epsilon}\log \la)=\la^{-1}\exp(\frac{x}{M\epsilon}\log \la),
\end{eqnarray*}
and
\begin{eqnarray*}
-\frac{x}{M\epsilon}\la\exp(\frac{x}{M\epsilon}\log \la)=\d _{\la^{-1}}\exp(\frac{x}{M\epsilon}\log \la),
\end{eqnarray*}
we can prove the equations in \eqref{wRlinear}.
\end{proof}

Moreover, on the space of wave functions $w_L(t,\la)$ and $w_R(t,\la)$,
identities $[\L,M_L]=1$, $[\L,M_R]=1$ and $[\la,\partial_\la]=-1$ induce an
anti-isomorphism between $(\L,M_L)$, $(\L,M_R)$ and $(\la,\partial_\la)$,
i.e.,
\begin{eqnarray}\label{generalLM}
M_L^m\L^lw_L(t,\la)&=&\la^l\left(\partial_\la^mw_L(t,\la)\right),\\
\L^lM_L^mw_L(t,\la)&=&\partial_\la^m\left(\la^l w_L(t,\la)\right), m,l\in
Z_+.
\\
\label{generalLM+1}
M_R^m\L^lw_R(t,\la)&=&\la^{-l}\left(\partial_{\la^{-1}}^mw_R(t,\la)\right),\;\\
\L^lM_R^mw_R(t,\la)&=&\partial_{\la^{-1}}^m\left({\la}^{-l}
w_R(t,\la)\right), m,l\in Z_+.
\end{eqnarray}

\section{Additional symmetries of BTH}
We are now to define the additional flows, and then to
prove that they are symmetries, which are called additional
symmetries of the BTH. We introduce additional
independent variables $t^*_{m,l}$ and define the actions of the
additional flows on the wave operators as
\begin{eqnarray}\label{definitionadditionalflowsonphi2}
\dfrac{\partial \P_L}{\partial
{t^*_{m,l}}}=-\left((M_L-M_R)^m\L^l\right)_{-}\P_L, \ \ \ \dfrac{\partial
\P_R}{\partial {t^*_{m,l}}}=\left((M_L-M_R)^m\L^l\right)_{+}\P_R,
\end{eqnarray}
where $m\geq 0, l\geq 0$. The following theorem shows that the definition \eqref{definitionadditionalflowsonphi2} is compatible with reduction condition \eqref{constraint lax} of the  BTH.
\begin{theorem}\label{preserve constraint}
The additional flows \eqref{definitionadditionalflowsonphi2} preserve reduction condition \eqref{constraint lax}.
\end{theorem}
\begin{proof} By performing the derivative on $\L$ dressed by $\P_L$ and
using the additional flow about $\P_L$ in \eqref{definitionadditionalflowsonphi2}, we get
\begin{eqnarray*}
(\partial_{t^*_{m,l}}\L)&=& (\partial_{t^*_{m,l}}\P_L)\ \La^N \P_L^{-1}
+ \P_L\ \La^N\ (\partial_{t_{m,l}}\P_L^{-1})\\
&=&-((M_L-M_R)^m\L^l)_{-} \P_L\ \La^N\ \P_L^{-1}- \P_L\ \La^N
\P_L^{-1}\ (\partial_{t^*_{m,l}}\P_L)
\ \P_L^{-1}\\
&=&-((M_L-M_R)^m\L^l)_{-} \L+ \L ((M_L-M_R)^m\L^l)_{-}\\
&=&-[((M_L-M_R)^m\L^l)_{-},\L].
\end{eqnarray*}
Similarly, we perform the derivative on $\L$ dressed by $\P_R$ and
use the additional flow about $\P_R$ in \eqref{definitionadditionalflowsonphi2} to get the following
\begin{eqnarray*}
(\partial_{t^*_{m,l}}\L)&=& (\partial_{t^*_{m,l}}\P_R)\ \La^{-M} \P_R^{-1}
+ \P_R\ \La^{-M}\ (\partial_{t_{m,l}}\P_R^{-1})\\
&=&((M_L-M_R)^m\L^l)_{+} \P_R\ \La^{-M}\ \P_R^{-1}- \P_R\ \La^{-M}
\P_R^{-1}\ (\partial_{t^*_{m,l}}\P_R)
\ \P_R^{-1}\\
&=&((M_L-M_R)^m\L^l)_{+} \L- \L ((M_L-M_R)^m\L^l)_{+}\\
&=&[((M_L-M_R)^m\L^l)_{+},\L].
\end{eqnarray*}
Because
\begin{eqnarray}\label{ETHadditionalflow111.}
[M_L-M_R,\L]=0,
\end{eqnarray}
therefore
\begin{eqnarray}\label{ETHadditionalflow1111}
\dfrac{\partial \L}{\partial
{t^*_{m,l}}}=[-\left((M_L-M_R)^m\L^l\right)_{-},
\L]=[\left((M_L-M_R)^m\L^l\right)_{+}, \L],
\end{eqnarray}
which gives the compatibility of additional flow of BTH with reduction condition \eqref{constraint lax}.
\end{proof}

Similarly, we can take derivatives on dressing structure of  $M_L$ and  $M_R$ to get the following proposition.
\begin{proposition}\label{add flow}
The additional derivatives  act on  $M_L$, $M_R$ as
\begin{eqnarray}
\label{ETHadditionalflow11'}
\dfrac{\partial
M_L}{\partial{t^*_{m,l}}}&=&[-\left((M_L-M_R)^m\L^l\right)_{-}, M_L],
\\
\label{ETHadditionalflow12}
\dfrac{\partial
M_R}{\partial{t^*_{m,l}}}&=&[\left((M_L-M_R)^m\L^l\right)_{+}, M_R].
\end{eqnarray}
\end{proposition}
\begin{proof} By performing the derivative on  $M_L$ given in (\ref{Moperator}), there exists
a similar derivative as $\partial_{t^*_{m,l}}\L$, i.e.,
\begin{eqnarray*}
(\partial_{t^*_{m,l}}M_L)&\!\!\!=\!\!\!&(\partial_{t^*_{m,l}}\P_L)\ \Gamma_L \P_L^{-1}
+ \P_L\ \Gamma_L\ (\partial_{t^*_{m,l}}\P_L^{-1})\\
&\!\!\!=\!\!\!&-((M_L-M_R)^m\L^l)_{-} \P_L\ \Gamma_L\ \P_L^{-1}- \P_L\ \Gamma_L
\P_L^{-1}\ (\partial_{t^*_{m,l}}\P_L)
\ \P_L^{-1}\\
&\!\!\!=\!\!\!&-((M_L-M_R)^m\L^l)_{-} M_L+ M_L
((M_L-M_R)^m\L^l)_{-}\\
&=&-[((M_L-M_R)^m\L^l)_{-}, M_L].
\end{eqnarray*}
Here the fact that $\Gamma_L$ does not depend on the additional
variables $t^*_{m,l}$ has been used. Other identities can also
be obtained in a similar way.
\end{proof}

By two propositions above,  the following corollary can be easily got.
\begin{corollary}\label{additionflowsonLnMmAnk}
For $1\leq \alpha \leq N, -M+1\leq \beta \leq 0, \ \ n,m,l\geq 0$,
the following identities hold
\begin{eqnarray}\label{ETHadditionalflow4}
&\!\!\!\!\!\!\!\!\!\!\!\!\!\!\!\!\!\!\!\!&\dfrac{\partial \L^{\frac nN}}{\partial{t^*_{m,l}}}=[-((M_L-M_R)^m\L^l)_{-},
\L^{\frac nN}] ,\ \ \ \ \ \ \ \ \
 \dfrac{\partial
B_{\alpha,n}}{\partial{t^*_{m,l}}}=[-((M_L-M_R)^m\L^l)_{-}, B_{\alpha,n}],
\\&\!\!\!\!\!\!\!\!\!\!\!\!\!\!\!\!\!\!\!\!&
\label{ETHadditionalflow4'}
\dfrac{\partial \L^{\frac nM}}{\partial{t^*_{m,l}}}=[((M_L-M_R)^m\L^l)_{+},
\L^{\frac nM}] ,\ \ \ \ \ \ \ \ \ \ \
 \dfrac{\partial
B_{\beta,n}}{\partial{t^*_{m,l}}}=[((M_L-M_R)^m\L^l)_{+}, B_{\beta,n}],
\\&\!\!\!\!\!\!\!\!\!\!\!\!\!\!\!\!\!\!\!\!&
\label{ETHadditionalflow5}
 \dfrac{\partial M_L^n}{\partial{t^*_{m,l}}}=[-((M_L-M_R)^m\L^l)_{-}, M_L^n]
,\ \ \ \ \ \ \ \ \
 \dfrac{\partial M_R^n}{\partial{t^*_{m,l}}}=[((M_L-M_R)^m\L^l)_{+}, M_R^n],
\\&\!\!\!\!\!\!\!\!\!\!\!\!\!\!\!\!\!\!\!\!&
\label{eqadditionflowsonLnMmAnk'}
 \dfrac{\partial M_L^n\L^k}{\partial{t^*_{m,l}}}=-[((M_L-M_R)^m\L^l)_{-}, M_L^n\L^k]
,\ \ \
 \dfrac{\partial M_R^n\L^k}{\partial{t^*_{m,l}}}=[((M_L-M_R)^m\L^l)_{+},
 M_R^n\L^k].
\end{eqnarray}
\end{corollary}
\begin{proof}  First we present the proof of the first
equation.   Considering the dressing relation
\begin{eqnarray*}
 \L^{\frac nN}=\P_L\Lambda^n\P_L^{-1},
\end{eqnarray*}
and \eqref{definitionadditionalflowsonphi2}, we can get relations
\begin{eqnarray*}
\dfrac{\partial \L^{\frac{n}{N}}}{\partial{t^*_{m,l}}}=[-((M_L-M_R)^m\L^l)_{-}, \L^{\frac{n}{N}}],
\end{eqnarray*}
which further leads to the second identity in
\eqref{ETHadditionalflow4}. Similarly, by relations
\begin{eqnarray*}
 \L^{\frac nM}=
 \P_R\Lambda^{-n}\P_R^{-1},\ \ \ \ 
 M_L=\P_L\Gamma_L\P_L^{-1},\ \ \ \ 
 M_R=\P_R\Gamma_R\P_R^{-1},
\end{eqnarray*}
and \eqref{definitionadditionalflowsonphi2}, we can get other
relations in similar ways.
\end{proof}

With Proposition \ref{add flow} and Corollary \ref{additionflowsonLnMmAnk}, the following theorem can be proved.

\begin{theorem}\label{symmetry}
The additional flows $\partial_{t^*_{m,l}}$ commute
with the bigraded  Toda hierarchy flows $\partial_{t_{\gamma,n}}$, i.e.,
\begin{eqnarray}
[\partial_{t^*_{m,l}}, \partial_{t_{\gamma,n}}]\Phi=0,
\end{eqnarray}
where $\Phi$ can be $\P_L$, $\P_R$ or $L$, $-M+1\leq \gamma \leq N$ and
 $
\partial_{t^*_{m,l}}=\frac{\partial}{\partial{t^*_{m,l}}},
\partial_{t_{\gamma,n}}=\frac{\partial}{\partial{t_{\gamma,n}}}$.

\end{theorem}
\begin{proof} According to the definition,
\begin{eqnarray*}
[\partial_{t^*_{m,l}},\partial_{t_{\gamma,n}}]\P_L=\partial_{t^*_{m,l}}
(\partial_{t_{\gamma,n}}\P_L)-
\partial_{t_{\gamma,n}} (\partial_{t^*_{m,l}}\P_L),
\end{eqnarray*}
and using the actions of the additional flows and the bigraded Toda
flows on $\P_L$, for  $1\leq\alpha\leq N$, we have
\begin{eqnarray*}
[\partial_{t^*_{m,l}},\partial_{t_{\alpha,n}}]\P_L
&=& -\partial_{t^*_{m,l}}\left((B_{\alpha,n})_{-}\P_L\right)+
\partial_{t_{\alpha,n}} \left(((M_L-M_R)^m\L^l)^m_{-}\P_L \right)\\
&=& -(\partial_{t^*_{m,l}}B_{\alpha,n} )_{-}\P_L-
(B_{\alpha,n})_{-}(\partial_{t^*_{m,l}}\P_L)\\&&+
[\partial_{t_{\alpha,n}} ((M_L-M_R)^m\L^l)]_{-}\P_L +
((M_L-M_R)^m\L^l)_{-}(\partial_{t_{\alpha,n}}\P_L).
\end{eqnarray*}
Using \eqref{definitionadditionalflowsonphi2} and Theorem \ref{flowsofM}, it
equals
\begin{eqnarray*}
[\partial_{t^*_{m,l}},\partial_{t_{\alpha,n}}]\P_L
&=&[\left((M_L-M_R)^m\L^l\right)_{-}, B_{\alpha,n}]_{-}\P_L+
(B_{\alpha,n})_{-}\left((M_L-M_R)^m\L^l\right)_{-}\P_L\\
&&+[(B_{\alpha,n})_{+},(M_L-M_R)^m\L^l]_{-}\P_L-((M_L-M_R)^m\L^l)_{-}(B_{\alpha,n})_{-}\P_L\\
&=&[((M_L-M_R)^m\L^l)_{-}, B_{\alpha,n}]_{-}\P_L- [(M_L-M_R)^m\L^l,
(B_{\alpha,n})_{+}]_{-}\P_L\\&&+
[(B_{\alpha,n})_{-},((M_L-M_R)^m\L^l)_{-}]\P_L\\
&=&0.
\end{eqnarray*}
Similarly, using \eqref{definitionadditionalflowsonphi2} and Theorem \ref{flowsofM}, we can prove the additional flows   commute with
flows $\d_{t_{\beta,n}}$ with $-M+1\leq\beta\leq 0$ in the sense
of acting on $\P_L$. Of course, using
\eqref{definitionadditionalflowsonphi2}, \eqref{bn1},
\eqref{bn1'} and Theorem \ref{flowsofM}, we can also prove the additional flows  commute
with all flows of BTH  in the sense of acting on $\P_R, L$. Here
we also give the proof for commutativity of additional symmetries with
$\d_{t_{\beta,n}}$, where $-M+1\leq\beta\leq 0$. To be a little
different from the proof above, we let the Lie bracket act on $\P_R$,
\begin{eqnarray*}
[\partial_{t^*_{m,l}},\partial_{t_{\beta,n}}]\P_R &=&
\partial_{t^*_{m,l}}\left((B_{\beta,n})_{+}\P_R \right)-
\partial_{t_{\beta,n}} \left(((M_L-M_R)^m\L^l)_{+}\P_R  \right)\\
&=& (\partial_{t^*_{m,l}}B_{\beta,n} )_{+}\P_R +
(B_{\beta,n})_{+}(\partial_{t^*_{m,l}}\P_R
)\\&&-(\partial_{t_{\beta,n}} ((M_L-M_R)^m\L^l))_{+}\P_R
-((M_L-M_R)^m\L^l)_{+}(\partial_{t_{\beta,n}}\P_R ),
\end{eqnarray*}
which further leads to
\begin{eqnarray*}
[\partial_{t^*_{m,l}},\partial_{t_{\beta,n}}]\P_R
&=&[((M_L-M_R)^m\L^l)_{+}, B_{\beta,n}]_{+}\P_R+
(B_{\beta,n})_{+}((M_L-M_R)^m\L^l)_{+}\P_R \\
&&+[(B_{\beta,n})_{-},(M_L-M_R)^m\L^l]_{+}\P_R -((M_L-M_R)^m\L^l)_{+}(B_{\beta,n})_{+}\P_R \\
&=&[((M_L-M_R)^m\L^l)_{+}, B_{\beta,n}]_{+}\P_R
+[(B_{\beta,n})_{-},(M_L-M_R)^m\L^l]_{+}\P_R\\&& +
[(B_{\beta,n})_{+},((M_L-M_R)^m\L^l)_{+}]\P_R \\
&=&[((M_L-M_R)^m\L^l)_{+}, B_{\beta,n}]_{+}\P_R  +
[B_{\beta,n},((M_L-M_R)^m\L^l)_{+}]_{+}\P_R =0.
\end{eqnarray*}
In the proof above, $[(B_{\gamma,n})_{+},
((M_L-M_R)^m\L^l)]_{-}= [(B_{\gamma,n})_{+}, ((M_L-M_R)^m\L^l)_{-}]_{-}$ has
been used, since $[(B_{\gamma,n})_{+}, ((M_L-M_R)^m\L^l)_{+}]_{-}=0$.
\end{proof}
The commutative property in Theorem \ref{symmetry} means that
additional flows are symmetries of the BTH.
Since they are symmetries, it is natural to consider the algebraic
structures among these additional symmetries. So we obtain the following important
theorem.
\begin{theorem}\label{WinfiniteCalgebra}
The additional flows  $\partial_{t^*_{m,l}}(m>0,l\geq 0)$ form a Block type Lie algebra with the
following relation
 \begin{eqnarray}\label{algebra relation}
[\partial_{t^*_{m,l}},\partial_{t^*_{n,k}}]= (km-n l)\d^*_{m+n-1,k+l-1},
\end{eqnarray}
which holds in the sense of acting on  $\P_L$, $\P_R$ or $\L$ and  $m,n\geq0,\, l,k\geq 0.$
\end{theorem}
\begin{proof}
 By using
 (\ref{definitionadditionalflowsonphi2}), we get
\begin{eqnarray*}
[\partial_{t^*_{m,l}},\partial_{t^*_{n,k}}]\P_L&=&
\partial_{t^*_{m,l}}(\partial_{t^*_{n,k}}\P_L)-
\partial_{t^*_{n,k}}(\partial_{t^*_{m,l}}\P_L)\\
&=&-\partial_{t^*_{m,l}}\left(((M_L-M_R)^n\L^k)_{-}\P_L\right)
+\partial_{t^*_{n,k}}\left(((M_L-M_R)^m\L^l)_{-}\P_L\right)\\
&=&-(\partial_{t^*_{m,l}}
(M_L-M_R)^n\L^k)_{-}\P_L-((M_L-M_R)^n\L^k)_{-}(\partial_{t^*_{m,l}} \P_L)\\
&&+ (\partial_{t^*_{n,k}} (M_L-M_R)^m\L^l)_{-}\P_L+
((M_L-M_R)^m\L^l)_{-}(\partial_{t^*_{n,k}} \P_L).
\end{eqnarray*}
On the account of  Proposition\ref{add flow}, we further get
 \begin{eqnarray*}&&
[\partial_{t^*_{m,l}},\partial_{t^*_{n,k}}]\P_L\\
&=&-\Big[\sum_{p=0}^{n-1}
(M_L-M_R)^p(\partial_{t^*_{m,l}}(M_L-M_R))(M_L-M_R)^{n-p-1}\L^k
+(M_L-M_R)^n(\partial_{t^*_{m,l}}\L^k)\Big]_{-}\P_L\\&&-((M_L-M_R)^n\L^k)_{-}(\partial_{t^*_{m,l}} \P_L)\\
&&+\Big[\sum_{p=0}^{m-1}
(M_L-M_R)^p(\partial_{t^*_{n,k}}(M_L-M_R))(M_L-M_R)^{m-p-1}\L^l
+(M_L-M_R)^m(\partial_{t^*_{n,k}}\L^l)\Big]_{-}\P_L\\&&+
((M_L-M_R)^m\L^l)_{-}(\partial_{t^*_{n,k}} \P_L)\\
&=&[(nl-km)(M_L-M_R)^{m+n-1}\L^{k+l-1}]_-\P_L\\
&=&(km-nl)\d^*_{m+n-1,k+l-1}\P_L.
\end{eqnarray*}
Similarly  the same results on $\P_R$ and $\L$ are as follows
 \begin{eqnarray*}
[\partial_{t^*_{m,l}},\partial_{t^*_{n,k}}]\P_R
&=&((km-nl)(M_L-M_R)^{m+n-1}\L^{k+l-1})_+\P_R\\
&=&(km-nl)\d^*_{m+n-1,k+l-1}\P_R,
\\[6pt]
{}[\partial_{t^*_{m,l}},\partial_{t^*_{n,k}}]\L&=&
\partial_{t^*_{m,l}}(\partial_{t^*_{n,k}}\L)-
\partial_{t^*_{n,k}}(\partial_{t^*_{m,l}}\L)\\
&=&[((nl-km)(M_L-M_R)^{m+n-1}\L^{k+l-1})_-, \L]\\
&=&(km-nl)\d^*_{m+n-1,k+l-1}\L.
\end{eqnarray*}
\end{proof}
Denote $d_{m,l}=\d_{t^*_{m+1,l+1}}$, and let $\BB$ be the span of all $d_{m,l},\,m,l\ge-1$.
Then by \eqref{algebra relation}, $\BB$ is a Lie algebra with relations
\begin{eqnarray}
[d_{m,l},d_{n,k}]=((m+1)(k+1)-(l+1)(n+1))d_{m+n,l+k}, \ \mbox{ \ for \ } m,n,l,k\geq -1.
\end{eqnarray}
Thus $\BB$ is in fact a Block type Lie algebra [\ref{Block}--\ref{Su}] (or more precisely the
upper infinite part of a Block type Lie algebra \cite{Block}). We see that
$\BB$ is generated by the set
\begin{eqnarray}\label{Gen-B}
B=\{ d_{-1,0},
d_{0,-1}, d_{0,0}, d_{1,0}, d_{0,1}\}=\{ \d^*_{0,1}=\d_{0,0}=\d_{1,0}, \d^*_{1,0}, \d^*_{1,1}, \d^*_{2,1},
\d^*_{1,2}\}.
\end{eqnarray}
%
%
%
%
It is easily to see that
$\{d_{m,0}, m\ge -1\}$ and $\{d_{0,m}, m\ge-1\} $ both span half of the centerless Virasoro algebra or Witt algebra.
It is a challenging work to get some explicit representations of general Block type Lie algebras. Actually, as we shall show
in later sections  it is highly nontrivial to do this even for  $\BB$.

To get the intuitive understanding of the properties of the additional symmetries of the BTH, we
would like to  give the first few flows of additional symmetry  in the following. These examples show that
additional symmetry flows indeed explicitly depend on the coordinate variables $t_{\gamma,n}$ and $x$.
For the $t^*_{1,0}$ flow,  \eqref{ETHadditionalflow1111} implies
\begin{eqnarray*}
\dfrac{\partial
\L}{\partial{t^*_{1,0}}}&=&[-\left(M_L-M_R\right)_{-}, \L]\\
&=&1+\sum_{n> 0}\sum_ {\alpha=1}^{N}(n+1-\frac{\alpha-1}{N})t_{\alpha, n}\d_{\alpha,
n-1}\L+\sum_{n>0}\sum_ {\beta=-M+1}^{0}(n+1+\frac{\beta}{M})t_{\beta, n}\d_{\beta,
n-1}\L,
\end{eqnarray*}
which further leads to
\begin{eqnarray*}
\begin{cases}
\dfrac{\partial u_{i}}{\partial{t^*_{1,0}
}} &=-\sum\limits_{n> 0}\sum\limits_
{\alpha=1}^{N}(n+1-\frac{\alpha-1}{N})t_{\alpha, n}\d_{\alpha,n-1}u_{i}-\sum\limits_{n> 0}
\sum\limits_ {\beta=-M+1}^{0}(n+1+\frac{\beta}{M})t_{\beta, n}\d_{\beta, n-1}u_{i},\ \ \ i\neq 0,\\
\dfrac{\partial u_{0}}{\partial{t^*_{1,0}
}} &=1-\sum\limits_{n> 0}\sum\limits_
{\alpha=1}^{N}(n+1-\frac{\alpha-1}{N})t_{\alpha, n}\d_{\alpha,n-1}u_{0}-\sum\limits_{n> 0}\sum\limits_
{\beta=-M+1}^{0}(n+1+\frac{\beta}{M})t_{\beta, n}\d_{\beta, n-1}u_{0}.
\end{cases}
\end{eqnarray*}
Similarly, after some computations, we get the $t^*_{1,1}$ flow
\begin{eqnarray*}
\dfrac{\partial
\L}{\partial{t^*_{1,1}
}}
&=&\sum_{i\geq 1}^{N+M}\frac{i}{N}u_{N-i}\Lambda^{N-i}-\sum_{n\geq 0}\sum_
{\alpha=1}^{N}(n+1-\frac{\alpha-1}{N})t_{\alpha,
n}\d_{\alpha,n}\L\\
&&-\sum_{n\geq 0}\sum_
{\beta=-M+1}^{0}(n+1+\frac{\beta}{M})t_{\beta, n}\d_{\beta, n}\L,
\end{eqnarray*}
which leads to
 \begin{eqnarray*}
\dfrac{\partial u_{i}}{\partial{t^*_{1,1}
}}
&=&\frac{N-i}{N}u_{i}-\sum_{n\geq 0}\sum_
{\alpha=1}^{N}(n+1-\frac{\alpha-1}{N})t_{\alpha,
n}\d_{\alpha,n}u_{i}-\sum_{n\geq 0}\sum_
{\beta=-M+1}^{0}(n+1+\frac{\beta}{M})t_{\beta, n}\d_{\beta, n}u_{i},
\end{eqnarray*} for $-M\leq i\leq N-1.$

Other equations of additional symmetry can be got in similar ways which will not be mentioned here.


\sectionnew{ Block type actions  on  functions $P_L$ and  $P_R$}

In this section,  we shall
give some specific additional flows acting on wave operators $\P_L$, $\P_R$
and then on  functions $P_L$, $P_R$ which denote the symbol of $\P_L$and $\P_R$ respectively
(see \eqref{PLsymbol} and \eqref{PRsymbol}).

As we all know, W-constraint which means  additional flows of wave function become vanishing is a very useful tool to bring  action of additional flow on wave functions into action on tau function. But for the BTH, it is not easy to get the action on tau function. Even Virasoro constraint is still difficult to get which will be shown in the last section of this paper.  So we choose another  kind of constraints which are weaker than the W-constraints to give representations of the BTH. It is named weak W-constraints which will be shown later. By these specific actions on functions $P_L$, $P_R$,
we will present two representations of the Block type Lie algebra $\BB$ under this kind of so-called weak W-constraints.

By definition, then
\begin{eqnarray}
\dfrac{\partial
\P_L}{\partial{t^*_{0,1}}}&=&-\L_{-}\P_L= \partial_{0,0}\P_L,
\end{eqnarray}
which leads to
\begin{eqnarray}
\dfrac{\partial
P_L}{\partial_{t^*_{0,1}}}&=&A^*_{L0,1}P_L,
\mbox{ \ where}\notag\\
\label{L0,1}
A^*_{L0,1}&=& \partial_{0,0}.
\end{eqnarray}
Considering
\eqref{definitionadditionalflowsonphi2},
 the $t^*_{1,0}$ flow can be written as following
\begin{eqnarray*}
\dfrac{\partial
\P_L}{\partial{t^*_{1,0}}}&=&-\left(M_L-M_R\right)_{-}\P_L\\
&=&-\P_L\frac{x}{N\epsilon}\Lambda^{-N}-\sum_ {\alpha=2}^{N}t_{\alpha, 0}(1-\frac{\alpha-1}{N})\P_L\La^{1-\alpha}+\sum_{n> 0}\sum_ {\alpha=1}^{N}(n+1-\frac{\alpha-1}{N})t_{\alpha, n}\d_{\alpha,
n-1}\P_L
\\&&+\sum_{n>0}\sum_ {\beta=-M+1}^{0}(n+1+\frac{\beta}{M})t_{\beta, n}\d_{\beta,
n-1}\P_L.
\end{eqnarray*}
Because $[\P_L, \frac{x}{\epsilon}]z^{\frac{x}{\epsilon}}=z(\d_zP_L)z^{\frac{x}{\epsilon}}$,
\begin{eqnarray*}
\dfrac{\partial
P_L}{\partial{t^*_{1,0}}}z^{\frac{x}{\epsilon}}&=&\dfrac{\partial
\P_L}{\partial{t^*_{1,0}}}z^{\frac{x}{\epsilon}}\\
&=&\Big[-(\frac{z^{1-N}}{N}\d_zP_L) -\frac{x}{N\epsilon}z^{-N}P_L-\sum_ {\alpha=2}^{N}t_{\alpha, 0}(1-\frac{\alpha-1}{N})z^{1-\alpha}P_L\\&&+\sum_{n>0}\sum_ {\alpha=1}^{N}(n+1-\frac{\alpha-1}{N})t_{\alpha, n}(\d_{\alpha,
n-1}P_L)+\sum_{n>0}\sum_ {\beta=-M+1}^{0}(n+1+\frac{\beta}{M})t_{\beta, n}(\d_{\beta,
n-1}P_L)\Big]z^{\frac{x}{\epsilon}}.
\end{eqnarray*}
Then  the $t^*_{1,0}$ flow of  function $P_L$ is given as follows
\begin{eqnarray}
\dfrac{\partial
P_L}{\partial{t^*_{1,0}}}
&=&A^*_{L1,0}P_L,
\mbox{ \ where}\notag\\
\notag
A^*_{L1,0}
&=&-\frac{z^{1-N}}{N}\d_z -\frac{x}{N\epsilon}z^{-N}-\sum_ {\alpha=2}^{N}t_{\alpha, 0}(1-\frac{\alpha-1}{N})z^{1-\alpha}+\sum_{n>0}\sum_ {\alpha=1}^{N}(n+1-\frac{\alpha-1}{N})t_{\alpha, n}\d_{\alpha,
n-1}\\&&+\sum_{n>0}\sum_ {\beta=-M+1}^{0}(n+1+\frac{\beta}{M})t_{\beta, n}\d_{\beta,
n-1}.
\label{L1,0}
\end{eqnarray}
By the same calculation, we can get the $t^*_{1,1}$ flow of wave  operator $\P_L$ and function $P_L$ as follows
\begin{eqnarray}
\dfrac{\partial
\P_L}{\partial{t^*_{1,1}
}}
&\!\!\!=\!\!\!&-\P_L\frac{x}{N\epsilon}+\frac{x}{N\epsilon}\P_L-\sum_{n\geq
0}\sum_ {\alpha=1}^{N}(n+1-\frac{\alpha-1}{N})t_{\alpha,
n}\d_{\alpha,n}\P_L\notag\\
&&-\sum_{n\geq 0}\sum_ {\beta=-M+1}^{0}(n+1+\frac{\beta}{M})t_{\beta,
n}\d_{\beta, n}\P_L,
\notag\\[5pt]
%
%
\dfrac{\partial
P_L}{\partial{t^*_{1,1}
}}
&\!\!\!=\!\!\!&A^*_{L1,1}P_L,
\mbox{ \ where}\notag\\
\label{L1,1}
\!\!A^*_{L1,1}
&\!\!\!=\!\!\!&-\frac{z}{N}\d_z+\sum_{n\geq
0}\sum_ {\alpha=1}^{N}(n+1-\frac{\alpha-1}{N})t_{\alpha,
n}\d_{\alpha,n}+\sum_{n\geq 0}\sum_ {\beta=-M+1}^{0}(n+1+\frac{\beta}{M})t_{\beta,
n}\d_{\beta, n}.
\end{eqnarray}
Similarly, the $t^*_{1,2}$ flow of wave  operator $\P_L$ and  function $P_L$ can be got  as follows
\begin{eqnarray}
\!\!\!\dfrac{\partial
\P_L}{\partial{t^*_{1,2}
}}
&\!\!\!=\!\!\!&-\P_L\frac{x}{N\ep}\La^N+\sum_{i+j\leq N}\omega_i\La^{-i}\frac{x}{N\epsilon}\La^N\La^{-j}\omega'_j\P_L+\sum_{n\geq
0}\sum_
{\alpha=1}^{N}(n+1-\frac{\alpha-1}{N})t_{\alpha,
n}\d_{\alpha,
n+1}\P_L\notag\\
\!\!\!&&-\sum_{i+j< M}\tilde\omega_i\La^{i}\frac{x}{M\epsilon}\La^{-M}\La^{j}\tilde\omega'_j\P_L+\sum_{n\geq
0}\sum_ {\beta=-M+1}^{0}(n+1+\frac{\beta}{M})t_{\beta, n}\d_{\beta, n+1}\P_L,
\notag\\[5pt]
%
%
\!\!\!\dfrac{\partial
P_L}{\partial{t^*_{1,2}
}}
&\!\!\!=\!\!\!&A^*_{L1,2}P_L,
\mbox{ \ where}\notag\\
\notag
A^*_{L1,2}
&=&-\frac{z^{N+1}}{N}\d_z+(\frac{x}{N\ep}+\frac{x}{M\ep})\d_{t_{1,0}}-\frac{1}{N}\sum_{i+j\leq N}i\omega_iz^{N-i-j}\omega'_j(x+(N-i-j)\ep) \La^{N-i-j}\\ \notag
&&+\sum_{n\geq
0}\sum_
{\alpha=1}^{N}(n+1-\frac{\alpha-1}{N})t_{\alpha,
n}\d_{\alpha,
n+1}+\sum_{n\geq
0}\sum_ {\beta=-M+1}^{0}(n+1+\frac{\beta}{M})t_{\beta, n}\d_{\beta, n+1}\\ \label{L1,2'}
&&-\frac{1}{M}\sum_{i+j< M}i\tilde\omega_iz^{i+j-M}\tilde\omega'_j(x+(i+j-M)\ep) \La^{i+j-M}.
\end{eqnarray}
The general Virosora flows are
\begin{eqnarray}
\dfrac{\partial
P_L}{\partial{t^*_{1,m+1}
}}
&=&A^*_{L1,m+1}P_L, \ \ m=0,1,...,
\mbox{ \ where}\notag\\
%
\label{L1,m+1}
A^*_{L1,m+1}
&=&-\frac{z^{mN+1}}{N}\d_z+(\frac{x}{N\ep}+\frac{x}{M\ep})\d_{t_{1,m-1}}
\notag\\&&-
\frac{1}{N}\sum_{i+j\leq mN}i\omega_iz^{mN-i-j}\omega'_j(x+(m N-i-j)\ep) \La^{mN-i-j}\notag\\
&&+\sum_{n\geq
0}\sum_
{\alpha=1}^{N}(n+1-\frac{\alpha-1}{N})t_{\alpha,
n}\d_{\alpha,
n+m}+\sum_{n\geq
0}\sum_ {\beta=-M+1}^{0}(n+1+\frac{\beta}{M})t_{\beta, n}\d_{\beta, n+m}\notag\\&&-
\frac{1}{M}\sum_{i+j< mM}i\tilde\omega_iz^{i+j-mM}\tilde\omega'_j(x+(i+j-m M)\ep) \La^{i+j-mM}\notag.
\end{eqnarray}

Now we will consider the \emph{weak W-constraints}\footnote{The weak W-constraints mean that the coefficients $\{\om_i, \om'_i, \tilde \om_i, \tilde \om'_i,i\geq 0\}\ (\om_0=\om'_0=1)$ satisfy the vanishing property of  derivatives of all operators $A^*_{Lm,n}$ and  $A^*_{Rm,n}$(will be introduced later) with respect to additional time variables, i.e.
 \begin{eqnarray}\label{weakwconstraint}
\d^*_{k,s}A^*_{Lm,n}&=&\d^*_{k,s}A^*_{Rm,n}=0,
\end{eqnarray}
for all $k,s,m,n\geq 0$.  W-constraints can imply weak W-constraints. Therefore condition under weak W-constraints  is weaker than W-constraints.}.
For example, condition
 \begin{eqnarray}\label{weakwconstraint1}
\d^*_{k,s}A^*_{L1,2}=0,
\end{eqnarray}
will lead to the following constraint
 \begin{eqnarray}\label{weakwconstraint2}
&&\d^*_{k,s}\left(\sum_{i+j=l\leq N}i\omega_i\omega'_j(x+(N-l)\ep)\right)
=0\\
&&\d^*_{k,s}\left(\sum_{i+j=s< M}i\tilde\omega_i\tilde\omega'_j(x+(s-M)\ep)\right)=0,
\end{eqnarray}
where $k,s,l,s\geq 0$.

Considering the weak W-constraints \eqref{weakwconstraint}, these operators $\{A^*_{L1,m}, m\geq 0\}$ can be regarded as a representation of the Virasoro algebra acting on $P_L$ function space, i.e.,
\begin{eqnarray}
[A^*_{L1,m},A^*_{L1,n}]P_L=(m-n)A^*_{L1,m+n-1}P_L.
\end{eqnarray}

After a tedious calculation, we can get the $t^*_{2,1}$ flow  on  function $P_L$
\begin{eqnarray}\label{L2,1}
\dfrac{\partial
P_L}{\partial{t^*_{2,1}
}}=A^*_{L2,1}P_L,
\end{eqnarray}
where $A^*_{L2,1}$ has very complicated form which will be presented in \eqref{L2,1a} in the appendix.

Till now, we have given all generating elements \eqref{Gen-B} of the Block type Lie algebra $\BB$ acting on function $P_L$,  i.e. $\{A^*_{L0,1}, A^*_{L1,0}, A^*_{L1,1}, A^*_{L2,1}, A^*_{L1,2}\}$ in \eqref{L0,1}--\eqref{L2,1}.
Thus we have in fact given a representation of $\BB$ on $P_L$ function space.

Similarly  corresponding  flows of  $P_R$ can be got as follows.
Firstly,
by definition,
\begin{eqnarray*}
\dfrac{\partial
\P_R}{\partial{t^*_{0,1}}}&=&\L_{+}\P_R= \partial_{1,0}\P_R,
\end{eqnarray*}
which leads to
\begin{eqnarray}
\dfrac{\partial
P_R}{\partial_{t^*_{0,1}}}&=&A^*_{R0,1}P_R,
\mbox{ \ where}\notag\\
\label{R0,1}
A^*_{R0,1}&=& \partial_{1,0}.
\end{eqnarray}
Secondly, by the following calculation
\begin{eqnarray*}
\dfrac{\partial
\P_R}{\partial{t^*_{1,0}}}&=&\left(M_L-M_R\right)_{+}\P_R\\
&=&\P_R\frac{x}{M\epsilon}\Lambda^{M}+\sum_{n>
0}\sum_ {\alpha=1}^{N}(n+1-\frac{\alpha-1}{N})t_{\alpha, n}\d_{\alpha, n-1}\P_R+t_{1, 0}\P_R\\
&&+\sum_{n>
0}\sum_ {\beta=-M+1}^{0}(n+1+\frac{\beta}{M})t_{\beta, n}\d_{\beta, n-1}\P_R+\sum_ {\beta=-M+1}^{0}(1+\frac{\beta}{M})t_{\beta, 0}\P_R\Lambda^{-\beta},
\end{eqnarray*}
and considering  $[\P_R, \frac{x}{\epsilon}]z^{\frac{x}{\ep}}=z(\d_zP_R)z^{\frac{x}{\ep}}$, we can get

\begin{eqnarray}\label{R1,0}
\dfrac{\partial
P_R}{\partial{t^*_{1,0}}}&=&A^*_{R1,0}P_R,
\mbox{ \ where}\notag\\
A^*_{R1,0}
&=&\frac{z^{M+1}}{M}\d_z+\frac{x}{M\epsilon}z^M+\sum_{n>
0}\sum_ {\alpha=1}^{N}(n+1-\frac{\alpha-1}{N})t_{\alpha, n}\d_{\alpha, n-1}+t_{1, 0}\notag\\&&+\sum_{n>
0}\sum_ {\beta=-M+1}^{0}(n+1+\frac{\beta}{M})t_{\beta, n}\d_{\beta, n-1}+\sum_ {\beta=-M+1}^{0}(1+\frac{\beta}{M})t_{\beta, 0}z^{-\beta}.
\end{eqnarray}
In the same way,
\begin{eqnarray}\label{R1,1}
\dfrac{\partial
P_R}{\partial{t^*_{1,1}
}}
&=&A^*_{R1,1}P_R,
\mbox{ \ where}\notag\\
A^*_{R1,1}
&=&\frac{z}{M}\d_z+\frac{x}{M\epsilon}+\frac{x}{N\epsilon}+\sum_{n\geq
0}\sum_ {\alpha=1}^{N}(n+1-\frac{\alpha-1}{N})t_{\alpha,
n}\d_{\alpha,n}\notag\\
&&+\sum_{n\geq 0}\sum_ {\beta=-M+1}^{0}(n+1+\frac{\beta}{M})t_{\beta,
n}\d_{\beta, n}.
\end{eqnarray}
Similarly, the $t^*_{1,2}$  flows of $P_R$ can be derived as follows
\begin{eqnarray}\label{R1,2}
\dfrac{\partial
P_R}{\partial{t^*_{1,2}
}}
&=&A^*_{R1,2}P_R,
\mbox{ \ where}\notag\\
A^*_{R1,2}
&=&(\frac{x}{N\ep}+\frac{x}{M\ep})\d_{t_{1,0}}-\frac{1}{N}\sum_{i+j\leq N}i\omega_iz^{N-i-j}\omega'_j(x+(N-i-j)\ep) \La^{N-i-j}\notag\\
&&+\sum_{n\geq
0}\sum_
{\alpha=1}^{N}(n+1-\frac{\alpha-1}{N})t_{\alpha,
n}\d_{\alpha,
n+1}+\sum_{n\geq
0}\sum_{\beta=-M+1}^{0}(n+1+\frac{\beta}{M})t_{\beta, n}\d_{\beta, n+1}\notag\\
&&+\frac{1}{M}\sum_{i+j\geq M}i\tilde\omega_iz^{i+j-M}\tilde\omega'_j(x+(i+j-M)\ep) \La^{i+j-M}\notag.
\end{eqnarray}
Similarly, after a tedious calculation, we obtain the $t^*_{2,1}$ flow function $P_R$  as follows
\begin{eqnarray}\label{R2,1}
\dfrac{\partial
P_R}{\partial{t^*_{2,1}
}}=A^*_{R2,1}P_R,
\end{eqnarray}
where $A^*_{R2,1}$ has complicated form which will also be presented presented in \eqref{R2,1a} in the appendix.

Denote $A^*_{Lm+1,l+1}$ and $A^*_{Rm+1,l+1}$ by $D^*_{Lm,l}$ and $D^*_{Rm,l}$ respectively.
Considering the weak W-constraints  of  functions $P_L(t,z)$ and $P_R(t,z)$,
there exist two anti-isomorphisms between $\{d^*_{m,l}\}$ and $\{D^*_{Lm,l}\}$,  $\{d^*_{m,l}\}$ and $\{D^*_{Rm,l}\}$ respectively
as follows
\begin{eqnarray*}
P_L:\ \ \ \ \ \ \ \ \  d^*_{m,l}&\mapsto& D^*_{Lm,l}\\
\[d^*_{m,l},d^*_{n,k}\]&\mapsto& \[D^*_{Ln,k},D^*_{Lm,l}\],
\\[5pt]
P_R:\ \ \  \ \ \ \ \ d^*_{m,l}&\mapsto& D^*_{Rm,l}\\
\[d^*_{m,l},d^*_{n,k}\]&\mapsto& \[D^*_{Rn,k},D^*_{Rm,l}\].
\end{eqnarray*}
The anti-isomorphisms under Lie bracket are
\begin{eqnarray*}
\[d^*_{m,l},d^*_{n,k}\]P_L=\[D^*_{Ln,k},D^*_{Lm,l}\]P_L,
\mbox{ \ \ and \ \ }
\[d^*_{m,l},d^*_{n,k}\]P_R=\[D^*_{Rn,k},D^*_{Rm,l}\]P_R.
\end{eqnarray*}

We can  get other additional flows on wave operators and  functions by very tedious calculations using recursion relation, which will not be mentioned here. It would be great to get the actions of additional flows on $\tau$ function which can lead to the Virasoro constraints.
Although it is not easy to bring action onto tau functions under W-constraints, we can also get the following important theorem.
\begin{theorem}\label{Rep-Block}
\begin{itemize}\item[\rm(1)]Equations \eqref{L0,1}--\eqref{L2,1} give a representation of $\BB$ on $P_L$ function space.
\item[\rm(2)]Equations \eqref{R0,1}--\eqref{R2,1} give a representation of $\BB$ on $P_R$ function space.
\end{itemize}
\end{theorem}
It should be  remarked that the two representations in Theorem \ref{Rep-Block} are both under the weak W-constraints.
And also we would like to emphasize that  it is quite nontrivial to give representations of Block Lie algebra in Theorem \ref{Rep-Block} from the point of representation of infinite dimensional Lie algebra. It is difficult to give one representation of this Block Lie algebra but luckily we  get it from the first few additional flows of the BTH.

\sectionnew{Conclusions and discussions}
 We define Orlov-Schulman's $M_L$,
$M_R$ operators of BTH in Section 4 and give the  additional symmetries of the BTH in Section 5. Then
we find that the additional symmetries form one  kind of Block type
Lie algebra which has recently received much attention.
In Section 6, we obtain some flows on functions $P_L$ and $P_R$ which lead to two
representations of this Block type Lie algebra $\BB$ which contains the half centerless Virasoro algebra (or Witt
algebra) as a subalgebra.

The additional symmetries provide a very useful way to derive an explicit representations
of the Virasoro algebra for KP hierarchy \cite{os1}, BKP hierarchy
\cite{MingHsien Tu, MingHsien Tu2} and 2-D Toda hierarchy
\cite{adler95} by its action on the space of $\tau$ function. This representation is more elegant and simple
by comparing with its representation defined on the space of functions $P_L$, $P_R$.

Therefore we tried to get an explicit representation of the algebra
$\BB$ in a similar way through the actions of the
additional symmetries on the $\tau$ function of the BTH.
Due to some technical reasons, it does not seem easy to derive
additional actions on $\tau$ function under the  W-constraints, we
have thus only found the first few operators of the algebra
$\BB$. There are two ways to get the Virasoro
representations from the point of view of the integrable system.
 One way is to use ASvM formula \cite{adler95} and the other one is to use similar forms \cite{dickey1}.
Here we would like to use the second way to derive first two
operators in algebra $\BB$ from acting on $\tau$
function. Firstly, to this purpose, it is useful to rewrite its
action on the functions $P_L$, $P_R$ as a similar form \cite{dickey1}.
  Because the $\tau$ function of BTH is defined by $P_L$ as follows \cite{our paper}
  \begin{eqnarray}\label{pltau}P_L: &=&\frac{ \tau
(x, t-[z^{-1}]^N;\epsilon) }
     {\tau (x,t;\epsilon)},
     \end{eqnarray}
     where
     \begin{eqnarray} \notag
 \[z^{-1}\]^{N}_{\alpha,n} :=
\begin{cases}
  \frac{z^{-N(n+1-\frac{\alpha-1}{N})}}{N(n+1 - \frac{\alpha-1}{N}
)}, &\alpha=N,N-1,\dots 1,\\
 0, &\alpha = 0, -1\dots -(M-1),
  \end{cases}
\end{eqnarray}
a technical calculation  \cite{dickey1} can lead to the following identity in a similar form on the
$\tau$ function
\begin{eqnarray*}
\dfrac{\partial
P_L}{\partial{t^*_{1,0}
}}
&=&
A^*_{L1,0}P_L\\
&=&\Big[(\frac{x}{\epsilon}+\frac{N-1}{2})(\tilde t_{1,0}-t_{1,0})\\
&&+\sum_ {\alpha=2}^{N}\frac{(\alpha-1)(N-\alpha+1)}{2N}((t_{\alpha, 0}- \frac{z^{-(N-\alpha+1)}}{N-\alpha+1})(t_{N+2-\alpha, 0}-\frac{ z^{1-\alpha}}{\alpha-1})-
t_{\alpha, 0}t_{N+2-\alpha, 0})\\
&&+\sum_{n>0}\sum_ {\alpha=1}^{N}(n+1-\frac{\alpha-1}{N})(t_{\alpha, n}-\frac{z^{-N(n+1-\frac{\alpha-1}{N})}}{N(n+1 - \frac{\alpha-1}{N}
)})
\frac{\d_{\alpha,
n-1}\tilde \tau}{\tilde\tau}\\
&&-\sum_{n>0}\sum_ {\alpha=1}^{N}(n+1-\frac{\alpha-1}{N})t_{\alpha, n}\frac{\d_{\alpha,
n-1}\tau}{\tau}\\
&&+\sum_{n>0}\sum_ {\beta=-M+1}^{0}(n+1+\frac{\beta}{M})t_{\beta, n}(\frac{\d_{\beta,
n-1}\tilde \tau}{\tilde\tau}-\frac{\d_{\beta,
n-1} \tau}{\tau})\Big]\frac{\tilde\tau}{\tau},
\end{eqnarray*}
where $ \tau$ denotes $\tau(x, t, \ep)$ and $\tilde \tau$ denotes $\tau(x, t-[z^{-1}]^N, \ep)$.

 After supposing
 \begin{eqnarray*}
\dfrac{\partial
P_L}{\partial{t^*_{1,0}
}}
&=&0,
\end{eqnarray*}
and taking the integral constant (in fact this integral constant can depend on $t_{\beta,n}$ variables mentioned before) to be zero,
we can find
\begin{eqnarray*}
\L_{L1,0}\tau
&=&0,
\mbox{ \ where}\\
%
\L_{L1,0}
&=& (\frac{x}{\epsilon}+\frac{N-1}{2})t_{1,0}+\sum_ {\alpha=2}^{N}\frac{(\alpha-1)(N-\alpha+1)}{2N}
t_{\alpha, 0}t_{N+2-\alpha, 0}\\
&&+\sum_{n>0}\sum_ {\alpha=1}^{N}(n+1-\frac{\alpha-1}{N})t_{\alpha, n}\d_{\alpha,
n-1}+\sum_{n>0}\sum_ {\beta=-M+1}^{0}(n+1+\frac{\beta}{M})t_{\beta, n}\d_{\beta,
n-1}.
\end{eqnarray*}
Similarly, by the following calculation
\begin{eqnarray*}
\dfrac{\partial
P_L}{\partial{t^*_{1,1}
}}
&=&
A^*_{L1,1}P_L\\
&=&\Big[\sum_{n\geq
0}\sum_ {\alpha=1}^{N}(n+1-\frac{\alpha-1}{N})(t_{\alpha,
n}-\frac{z^{-N(n+1-\frac{\alpha-1}{N})} }{N(n+1 - \frac{\alpha-1}{N}
)})
\frac{\d_{\alpha,n}\tilde\tau}{\tilde \tau}\\
&&-
\sum_{n\geq
0}\sum_ {\alpha=1}^{N}(n+1-\frac{\alpha-1}{N})t_{\alpha,
n}\frac{\d_{\alpha,n}\tau}{ \tau}+\sum_{n\geq 0}\sum_ {\beta=-M+1}^{0}(n+1+\frac{\beta}{M})t_{\beta,
n}(\frac{\d_{\beta, n}\tilde\tau}{\tilde \tau}-\frac{\d_{\beta, n}\tau}{ \tau})\Big]\frac{\tilde\tau}{ \tau},
\end{eqnarray*}
 the second Virasoro operator can be got as following
\begin{eqnarray*}
\L_{L1,1}
&=&
\sum_{n\geq
0}\sum_ {\alpha=1}^{N}(n+1-\frac{\alpha-1}{N})t_{\alpha,
n}\d_{\alpha,n}+\sum_{n\geq 0}\sum_ {\beta=-M+1}^{0}(n+1+\frac{\beta}{M})t_{\beta,
n}\d_{\beta, n}.
\end{eqnarray*}
Only some of the Virasoro elements can be explicitly constructed because for the higher additional flows, similar forms as in \cite{dickey1} are too complicated. For example, we did not have the
  similar form of the $A^*_{L1,2}$ as in \eqref{L1,2'}. The above are about the actions on $P_L$. Similar situations
occur for the actions on $P_R$, which will not be mentioned here.
 Our future work will contain
ASvM formula for the BTH, which might  further lead to some representations of that Block type Lie algebra.
\vskip9pt

\small \noindent{\bf {Acknowledgments.}}
  {\small    This work is supported by the NSF of China under Grant No.~10971109, 10825101. Jingsong He is also supported by
Program for NCET under Grant No.NCET-08-0515. We thank Professor
Yishen Li (USTC, China) for long-term encouragements and supports.
Chuanzhong Li also thanks Professor Yuji Kodama (OSU, U.S.) and
Professor
 Hsianhua Tseng (OSU, U.S.) for their useful discussion. We also give thanks to referees for their valuable suggestions.}

\sectionnew{Appendix}

After a tedious calculation, we can get the $t^*_{2,1}$ flow  on wave function $P_L$,
\begin{eqnarray*}\label{L2,1'}
\dfrac{\partial
P_L}{\partial{t^*_{2,1}
}}=A^*_{L2,1}P_L,
\end{eqnarray*}
where,
{\small
\begin{eqnarray}
&&A^*_{L2,1}\label{L2,1a}\\&=&-\sum_{i=0}^{\infty}\omega_i\frac{x-i\ep}{N\epsilon}\frac{x-i\ep-N\ep}{N\epsilon}z^{-(N+i)}\La^{-(N+i)}\notag\\
&&-\sum_{n\geq 0}\sum_
{\alpha=1}^{N}\sum_{i+j >N(n-\frac{\alpha-1}{N})}
\omega_i(2\frac{x}{N\epsilon}+n-\frac{\alpha-1+2i}{N})z^{N(n-\frac{\alpha-1}{N})-i-j}\notag\\
&&
t_{\alpha, n}\omega'_j(x+[N(n-\frac{\alpha-1}{N})-i-j]\ep)\La^{N(n-\frac{\alpha-1}{N})-i-j}\notag\\
&&+\sum_{n,n'\geq 0}\sum_
{\alpha,\alpha'=1}^{N}(n+1-\frac{\alpha-1}{N})(n'+1-\frac{\alpha'-1}{N})
t_{\alpha, n}t_{\alpha', n'}\d_{\alpha+\alpha'-1,n+n'}\notag\\
&&
-\sum_{n\geq 0}\sum_
{\beta=-M+1}^{0}\sum_{i+j < M(n+\frac{\beta}{M})}(n+1+\frac{\beta}{M})
\tilde\omega_i(2\frac{x}{M\epsilon}-n-\frac{\beta-2i}{M})z^{-M(n+\frac{\beta}{M})+i+j}
t_{\beta, n}\notag\\
&&\tilde\omega'_j(x+[i+j-M(n+\frac{\beta}{M})]\ep)\La^{-M(n+\frac{\beta}{M})+i+j}\notag\\
&&+\sum_{n,n'\geq 0}\sum_
{\beta,\beta'=-M+1}^{0}(n+1+\frac{\beta}{M})(n'+1+\frac{\beta'}{M})
t_{\beta, n}t_{\beta', n'}\d_{\beta+\beta',n+n'}
\notag\\
&&
-\sum_{k+l< i+j+N}\om_i\frac{x-i\ep}{N\epsilon}\om'_j(x-(N+i+j)\ep)\tilde\om_k(x-(N+i+j)\ep)
\frac{x-(N+i+j-k)\ep}{M\epsilon}\notag\\
&&\tilde\om'_{l}(x+(k+l-N-i-j)\ep)z^{k+l-N-i-j}\La^{k+l-N-i-j}\notag\\
&&-\sum_{n\geq 0}\sum_
{\beta=-M+1}^{0}\sum_{k+l< i+j+N+M(n+1+\frac{\beta}{M})}\om_i\frac{x-i\ep}{N\epsilon}\om'_j(x-(N+i+j)\ep)\tilde\om_k(x-(N+i+j)\ep)\notag\\
&&
(n+1+\frac{\beta}{M})
z^{k+l-N-i-j-M(n+1+\frac{\beta}{M})}
t_{\beta, n}\tilde\om'_{l}(x+[k+l-N-i-j-M(n+1+\frac{\beta}{M})]\ep)\notag\\
&&\La^{k+l-N-i-j-M(n+1+\frac{\beta}{M})}\notag\\
&&-\sum_{n\geq 0}\sum_
{\alpha=1}^{N}\sum_{k+l< i+j-N(n-\frac{\alpha-1}{N})}\om_i(n+1-\frac{\alpha-1}{N})
z^{k+l+N(n-\frac{\alpha-1}{N})-i-j}\om'_j(x+[N(n-\frac{\alpha-1}{N})-i-j]\ep)\notag\\
&&\tilde\om_k
(x+[N(n-\frac{\alpha-1}{N})-i-j]\ep)\frac{x+[k+N(n-\frac{\alpha-1}{N})-i-j]\ep}{M\epsilon}\notag\\
&&\tilde\om'_{l}
(x+[k+l+N(n-\frac{\alpha-1}{N})-i-j]\ep)
\La^{k+l+N(n-\frac{\alpha-1}{N})-i-j}\notag\\
&&-\sum_{n,m\geq 0}\sum_
{\alpha=1}^{N}\sum_
{\beta=-M+1}^{0}(n+1-\frac{\alpha-1}{N})
(m+1+\frac{\beta}{M})
\sum_{k+l< i+j-N(n-\frac{\alpha-1}{N})+M(m+1+\frac{\beta}{M})}\notag\\
&&\om_i\om'_j
(x+[N(n-\frac{\alpha-1}{N})-i-j]\ep)\tilde\om_k(x+[N(n-\frac{\alpha-1}{N})-i-j]\ep)\notag\\
&&
z^{N(n-\frac{\alpha-1}{N})-i-j+k+l-M(m+1+\frac{\beta}{M})}t_{\alpha, n}
t_{\beta, m}\notag\\
&&\tilde\om'_{l}(x+[N(n-\frac{\alpha-1}{N})-i-j+k+l-M(m+1+\frac{\beta}{M})]\ep)\La^{N(n-\frac{\alpha-1}{N})-i-j+k+l-M(m+1+\frac{\beta}{M})}\notag\\
&&-\sum_{i+j+M< k+l}\tilde\om_i\frac{x+i\ep}{M\epsilon}\tilde\om'_j(x+(i+j+M)\ep)\om_k(x+(i+j+M)\ep)\notag\\
&&\frac{x+(i+j+M-k)\ep}{N\epsilon}\om'_{l}(x+(i+j+M-k-l)\ep)z^{i+j+M-k-l}
\La^{i+j+M-k-l}\notag\\
&&-\sum_{n\geq 0}\sum_
{\beta=-M+1}^{0}\sum_{i+j< k+l+M(n+\frac{\beta}{M})}\tilde\om_i(n+1+\frac{\beta}{M})
z^{i+j-M(n+\frac{\beta}{M})-k-l}
t_{\beta, n}\tilde\om'_j(x+[i+j-M(n+\frac{\beta}{M})]\ep)\notag\\
&&\om_k(x+[i+j-M(n+\frac{\beta}{M})]\ep)
\frac{x+[i+j-k-M(n+\frac{\beta}{M})]\ep}{N\epsilon}\notag\\
&&\om'_{l}(x+[i+j-k-l-M(n+\frac{\beta}{M})]\ep)\La^{i+j-k-l-M(n+\frac{\beta}{M})}\notag\\
&&-\sum_{n\geq 0}\sum_
{\alpha=1}^{N}\sum_{i+j+M< k+l -N(n+1+\frac{\alpha-1}{N})}\tilde\om_i\frac{x+i\ep}{M\epsilon}\tilde\om'_j
(x+(i+j+M)\ep)\om_k(x+(i+j+M)\ep)\notag\\
&&(n+1-\frac{\alpha-1}{N})
z^{N(n+1-\frac{\alpha-1}{N})-k-l+i+j+M}\om'_{l}(x+(i+j+M+N(n+1-\frac{\alpha-1}{N})-k-l)\ep)
\notag\\&&\La^{i+j+M+N(n+1-\frac{\alpha-1}{N})-k-l}-\sum_{n,m\geq 0}\sum_
{\alpha=1}^{N}\sum_
{\beta=-M+1}^{0}(n+1-\frac{\alpha-1}{N})
(m+1+\frac{\beta}{M})t_{\alpha, n}t_{\beta, m}\notag\\
&&
\sum_{k+l> i+j+N(n+1-\frac{\alpha-1}{N})-M(m+\frac{\beta}{M})}\tilde\om_i\tilde\om'_j
(x+[i+j-M(m+\frac{\beta}{M})]\ep)\om_k(x+[i+j-M(m+\frac{\beta}{M})]\ep)\notag\\
&&
z^{i+j-M(m+\frac{\beta}{M})+N(n+1-\frac{\alpha-1}{N})-k-l}
\om'_{l}(x+[i+j-M(m+\frac{\beta}{M})+N(n+1-\frac{\alpha-1}{N})-k-l]\ep)\notag\\
&&\La^{i+j-M(m+\frac{\beta}{M})+N(n+1-\frac{\alpha-1}{N})-k-l}.\notag
\end{eqnarray}
}
Similarly, after a tedious calculation, we can get the $t^*_{2,1}$ flow function $P_R$  as follows
\begin{eqnarray*}\label{R2,1'}
\dfrac{\partial
P_R}{\partial{t^*_{2,1}
}}=A^*_{R2,1}P_R,
\end{eqnarray*}
where,
{\small
\begin{eqnarray}\label{R2,1a}
&&A^*_{R2,1}\\
&=&\sum_{i=0}^{\infty}\tilde\omega_i\frac{x+i\ep}{M\epsilon}\frac{x+i\ep+M\ep}{M\epsilon}z^{M+i}\La^{M+i}\notag\\
&&+\sum_{n\geq 0}\sum_
{\alpha=1}^{N}\sum_{i+j\leq N(n-\frac{\alpha-1}{N})}(n+1-\frac{\alpha-1}{N})
\omega_i(2\frac{x}{N\epsilon}+n-\frac{\alpha-1+2i}{N})z^{N(n-\frac{\alpha-1}{N})-i-j}\notag\\
&&
t_{\alpha, n}\omega'_j(x+[N(n-\frac{\alpha-1}{N})-i-j]\ep)\La^{N(n-\frac{\alpha-1}{N})-i-j}\notag\\
&&+\sum_{n,n'\geq 0}\sum_
{\alpha,\alpha'=1}^{N}(n+1-\frac{\alpha-1}{N})(n'+1-\frac{\alpha'-1}{N})
t_{\alpha, n}t_{\alpha', n'}\d_{\alpha+\alpha'-1,n+n'}\notag\\
&&
+\sum_{n\geq 0}\sum_
{\beta=-M+1}^{0}\sum_{i+j \geq -M(n+\frac{\beta}{M})}(n+1+\frac{\beta}{M})
\tilde\omega_i(2\frac{x}{M\epsilon}-n-\frac{\beta-2i}{M})z^{-M(n+\frac{\beta}{M})+i+j}
t_{\beta, n}\notag\\
&&\tilde\omega'_j(x+[i+j-M(n+\frac{\beta}{M})]\ep)\La^{-M(n+\frac{\beta}{M})+i+j}\notag\\
&&+\sum_{n,n'\geq 0}\sum_
{\beta,\beta'=-M+1}^{0}(n+1+\frac{\beta}{M})(n'+1+\frac{\beta'}{M})
t_{\beta, n}t_{\beta', n'}\d_{\beta+\beta',n+n'}
\notag\\
&&
+\sum_{k+l\geq i+j+N}\om_i\frac{x-i\ep}{N\epsilon}\om'_j(x-(N+i+j)\ep)\tilde\om_k(x-(N+i+j)\ep)
\frac{x-(N+i+j-k)\ep}{M\epsilon}\notag\\
&&\om'_{l}(x+(k+l-N-i-j)\ep)z^{k+l-N-i-j}\La^{k+l-N-i-j}\notag\\
&&+\sum_{n\geq 0}\sum_
{\beta=-M+1}^{0}\sum_{k+l\geq i+j+N+M(n+1+\frac{\beta}{M})}\om_i\frac{x-i\ep}{N\epsilon}\om'_j(x-(N+i+j)\ep)\tilde\om_k(x-(N+i+j)\ep)\notag\\
&&
z^{k+l-N-i-j-M(n+1+\frac{\beta}{M})}
t_{\beta, n}\tilde\om'_{l}(x+[k+l-N-i-j-M(n+1+\frac{\beta}{M})]\ep)\notag\\
&&\La^{k+l-N-i-j-M(n+1+\frac{\beta}{M})}\notag\\
&&+\sum_{n\geq 0}\sum_
{\alpha=1}^{N}\sum_{k+l\geq i+j-N(n-\frac{\alpha-1}{N})}\om_i(n+1-\frac{\alpha-1}{N})
z^{k+l+N(n-\frac{\alpha-1}{N})-i-j}\om'_j(x+[N(n-\frac{\alpha-1}{N})-i-j]\ep)\notag\\
&&\tilde\om_k
(x+[N(n-\frac{\alpha-1}{N})-i-j]\ep)\frac{x+[k+N(n-\frac{\alpha-1}{N})-i-j]\ep}{M\epsilon}\notag\\
&&\tilde\om'_{l}
(x+[k+l+N(n-\frac{\alpha-1}{N})-i-j]\ep)
\La^{k+l+N(n-\frac{\alpha-1}{N})-i-j}\notag\\
&&+\sum_{n,m\geq 0}\sum_
{\alpha=1}^{N}\sum_
{\beta=-M+1}^{0}(n+1-\frac{\alpha-1}{N})
(m+1+\frac{\beta}{M})
\sum_{k+l\geq i+j-N(n-\frac{\alpha-1}{N})+M(m+1+\frac{\beta}{M})}\notag\\
&&\om_i\om'_j
(x+[N(n-\frac{\alpha-1}{N})-i-j]\ep)\tilde\om_k(x+[N(n-\frac{\alpha-1}{N})-i-j]\ep)\notag\\
&&
z^{N(n-\frac{\alpha-1}{N})-i-j+k+l-M(m+1+\frac{\beta}{M})}t_{\alpha, n}
t_{\beta, m}\notag\\
&&\tilde\om'_{l}(x+[N(n-\frac{\alpha-1}{N})-i-j+k+l-M(m+1+\frac{\beta}{M})]\ep)\La^{N(n-\frac{\alpha-1}{N})-i-j+k+l-M(m+1+\frac{\beta}{M})}\notag\\
&&+\sum_{i+j+M\geq k+l}\tilde\om_i\frac{x+i\ep}{M\epsilon}\tilde\om'_j(x+(i+j+M)\ep)\om_k(x+(i+j+M)\ep)\notag\\
&&\frac{x+(i+j+M-k)\ep}{N\epsilon}\om'_{l}(x+(i+j+M-k-l)\ep)z^{i+j+M-k-l}
\La^{i+j+M-k-l}\notag\\
&&+\sum_{n\geq 0}\sum_
{\beta=-M+1}^{0}\sum_{i+j\geq k+l+M(n+\frac{\beta}{M})}\tilde\om_i(n+1+\frac{\beta}{M})\notag\\
&&
z^{i+j-M(n+\frac{\beta}{M})-k-l}
t_{\beta, n}\tilde\om'_j(x+[i+j-M(n+\frac{\beta}{M})]\ep)\notag\\
&&\om_k(x+[i+j-M(n+\frac{\beta}{M})]\ep)
\frac{x+[i+j-k-M(n+\frac{\beta}{M})]\ep}{N\epsilon}\notag\\
&&\om'_{l}(x+[i+j-k-l-M(n+\frac{\beta}{M})]\ep)\La^{i+j-k-l-M(n+\frac{\beta}{M})}\notag\\
&&+\sum_{n\geq 0}\sum_
{\alpha=1}^{N}\sum_{i+j+M\geq k+l -N(n+1-\frac{\alpha-1}{N})}\tilde\om_i\frac{x+i\ep}{M\epsilon}\tilde\om'_j
(x+(i+j+M)\ep)\om_k(x+(i+j+M)\ep)\notag\\
&&(n+1-\frac{\alpha-1}{N})
z^{N(n+1-\frac{\alpha-1}{N})-k-l+i+j+M}\om'_{l}(x+(i+j+M+N(n+1-\frac{\alpha-1}{N})-k-l)\ep)
\notag\\&&\La^{i+j+M+N(n+1-\frac{\alpha-1}{N})-k-l}+\sum_{n,m\geq 0}\sum_
{\alpha=1}^{N}\sum_
{\beta=-M+1}^{0}(n+1-\frac{\alpha-1}{N})
(m+1+\frac{\beta}{M})t_{\alpha, n}t_{\beta, m}\notag\\
&&
\sum_{k+l\leq i+j+N(n+1-\frac{\alpha-1}{N})-M(m+\frac{\beta}{M})}\tilde\om_i\tilde\om'_j
(x+[i+j-M(m+\frac{\beta}{M})]\ep)\om_k(x+[i+j-M(m+\frac{\beta}{M})]\ep)\notag\\
&&
z^{i+j-M(m+\frac{\beta}{M})+N(n+1-\frac{\alpha-1}{N})-k-l}
\om'_{l}(x+[i+j-M(m+\frac{\beta}{M})+N(n+1-\frac{\alpha-1}{N})-k-l]\ep)\notag\\
&&\La^{i+j-M(m+\frac{\beta}{M})+N(n+1-\frac{\alpha-1}{N})-k-l}.\notag
\end{eqnarray}
}
\end{CJK*}
\end{document}